\documentclass[12pt,american,pra,reprint]{revtex4-2}
\usepackage[T1]{fontenc}
\usepackage[latin9]{inputenc}
\usepackage{color}
\usepackage{babel}
\usepackage{mathrsfs}
\usepackage{amsmath}
\usepackage{amsthm}
\usepackage{amssymb}
\usepackage{graphicx}
\usepackage[unicode=true,pdfusetitle,
 bookmarks=true,bookmarksnumbered=false,bookmarksopen=false,
 breaklinks=false,pdfborder={0 0 0},pdfborderstyle={},backref=false,colorlinks=true]
 {hyperref}

\makeatletter

\newcommand{\noun}[1]{\textsc{#1}}

\theoremstyle{plain}
\newtheorem{thm}{\protect\theoremname}
\theoremstyle{plain}
\newtheorem{lem}[thm]{\protect\lemmaname}
\theoremstyle{plain}
\newtheorem{prop}[thm]{\protect\propositionname}
\theoremstyle{definition}
\newtheorem{defn}[thm]{\protect\definitionname}
\theoremstyle{remark}
\newtheorem{rem}[thm]{\protect\remarkname}
\theoremstyle{definition}
\newtheorem{example}[thm]{\protect\examplename}

\usepackage{bbm}

\makeatother

\providecommand{\definitionname}{Definition}
\providecommand{\examplename}{Example}
\providecommand{\lemmaname}{Lemma}
\providecommand{\propositionname}{Proposition}
\providecommand{\remarkname}{Remark}
\providecommand{\theoremname}{Theorem}

\begin{document}
\title{Entanglement of a bipartite channel}
\author{Gilad Gour}
\email{gour@ucalgary.ca}

\affiliation{Department of Mathematics and Statistics, University of Calgary, AB,
Canada T2N 1N4}
\affiliation{Institute for Quantum Science and Technology, University of Calgary,
AB, Canada T2N 1N4}
\author{Carlo Maria Scandolo}
\email{carlomaria.scandolo@ucalgary.ca}

\affiliation{Department of Mathematics and Statistics, University of Calgary, AB,
Canada T2N 1N4}
\affiliation{Institute for Quantum Science and Technology, University of Calgary,
AB, Canada T2N 1N4}
\begin{abstract}
The most general quantum object that can be shared between two distant
parties is a bipartite channel, as it is the basic element to construct
all quantum circuits. In general, bipartite channels can produce entangled
states, and can be used to simulate quantum operations that are not
local. While much effort over the last two decades has been devoted
to the study of entanglement of bipartite states, very little is known
about the entanglement of bipartite channels. In this work, we rigorously
study the entanglement of bipartite channels as a resource theory
of quantum processes. We present an infinite and complete family of
measures of dynamical entanglement, which gives necessary and sufficient
conditions for convertibility under local operations and classical
communication. Then we focus on the dynamical resource theory where
free operations are positive partial transpose (PPT) superchannels,
but we do not assume that they are realized by PPT pre- and post-processing.
This leads to a greater mathematical simplicity that allows us to
express all resource protocols and the relevant resource measures
in terms of semi-definite programs. Along the way, we generalize the
negativity from states to channels, and introduce the max-logarithmic
negativity, which has an operational interpretation as the exact asymptotic
entanglement cost of a bipartite channel. Finally, we use the non-positive
partial transpose (NPT) resource theory to derive a no-go result:
it is impossible to distill entanglement out of bipartite PPT channels
under any sets of free superchannels that can be used in entanglement
theory. This allows us to generalize one of the long-standing open
problems in quantum information---the NPT bound entanglement problem---from
bipartite states to bipartite channels. It further leads us to the
discovery of bound entangled POVMs.
\end{abstract}
\maketitle

\section{Introduction}

Quantum entanglement \citep{Plenio-review,Review-entanglement} is
universally regarded as the most important aspect of quantum theory,
making it radically different from classical theory. Schrödinger himself
summarized this phenomenon as the fact that \citep{Schrodinger}
\begin{quotation}
``{[}\ldots{]} the best possible knowledge of a \emph{whole} does
not necessarily include the best possible knowledge of all its \emph{parts}.''
\end{quotation}
Indeed, entanglement is a necessary ingredient for the non-local phenomena
observed in quantum theory \citep{Bell,CHSH,Non-locality-review1,Non-locality-review2,LOSR-nonlocality}.
The development of quantum information theory has brought a new perspective
on quantum entanglement, seen as a resource in many protocols that
cannot be implemented in classical theory. Think, for instance, of
the paradigmatic examples of quantum teleportation \citep{Teleportation},
dense coding \citep{Dense-coding}, and quantum key distribution \citep{Ekert}.
The idea of entanglement concretely helping in information-theoretic
tasks can be made precise and rigorous using the framework of \emph{resource
theories} \citep{Quantum-resource-1,Quantum-resource-2,Resource-knowledge,Resource-currencies,Gour-single-shot,Regula2017,Multiresource,Gour-review,Adesso-resource,Single-shot-new,Kuroiwa2020generalquantum}.

Resource theories have been studied in great detail when the resources
involved are \emph{states} (also known as static resources) \citep{Gour-review}.
In this case, one wants to study the conversion between states. This
is the usual setting in which a rigorous theory of entanglement can
be put forward. The physical situation is when there are two separated
parties, and, because of their spatial separation, they are restricted
to performing local operations (LO), and exchanging Classical Communication
(CC) \citep{LOCC1,LOCC2,Lo-Popescu,Plenio-review,Review-entanglement}.
These free operations are called LOCC. In this setting, free states
are those that can be prepared from scratch using an LOCC protocol;
they are \emph{separable states}. Then one studies the conversion
between bipartite states when the two parties initially share a state,
which they are tasked to manipulate and transform into a target state
using LOCC channels. For pure entangled states, this conversion is
fairly easy to study \citep{Nielsen}, and for them the distillation
of maximal entangled states and the cost coincide. This is not the
case for entangled mixed states, for which the distillation can be
zero, yet the cost is strictly non-zero \citep{PHorodecki,Bound-entanglement}.
In other words, for some states, we need to spend maximally entangled
states to create them, but, once created, we cannot get back any maximal
entanglement. This phenomenon is called \emph{bound entanglement}.

Despite being the natural choice dictated by the physical setting
for entanglement, working with LOCC protocols is, in general, not
easy \citep{Chitambar2014}. For this reason, other choices of free
operations have been considered, which are structurally and mathematically
simpler to deal with. The first class is that of \emph{separable operations}
(SEP) \citep{SEP,Rains-SEP,PPT1}, which are the operations that send
separable states to separable states, even when tensored with the
identity channel. In resource-theoretic terminology they are \emph{completely
resource non-generating operations}, i.e.\ the largest set of free
operations (in the sense of inclusion) transforming free states into
free states, in a complete sense \citep{Gour-review}. LOCC channels
(and even their topological closure \citep{Chitambar2014}) have been
shown to be a \emph{strict} subset of separable operations \citep{SEP,DiVincenzo2003,Chitambar2009}.

We can also consider positive partial transpose (PPT) operations \citep{PPT1,PPT2}.
The definition of these operations is inspired by the Peres-Horodecki
criterion \citep{PPT-Peres,PPT-Horodecki} for the separability of
bipartite states, based on partial transpose: a state is separable
only if its partial transpose is still positive semi-definite. In
this resource theory, free states are states with positive semi-definite
partial transpose (\emph{PPT states}). They coincide with separable
states for bipartite systems of dimension $2\otimes2$ and $2\otimes3$,
but in general there are also non-separable PPT states \citep{PHorodecki}.
This is indeed the case for all known bound entangled states \citep{Bound-entanglement}.
In this non-positive partial transpose (NPT) resource theory, the
free operations are the channels that send PPT states to PPT states
even when tensored with the identity channel. They are called PPT
operations. Clearly both LOCC and separable operations are subsets
of PPT operations.

Despite not being so physically motivated, separable operations and
PPT operations are helpful for their greater mathematical simplicity,
and because they allow us to prove no-go results: if a state conversion
is \emph{not} possible under separable or PPT operations, then it
is \emph{not} possible under LOCC operations as well. Similarly, PPT
and separable operations can provide upper and lower bounds for conversions
with LOCC channels.

If one looks closely at the first examples where entanglement proved
to be a resource (e.g.\ quantum teleportation and dense coding),
one notices they involve the conversion of a state into a particular
channel, i.e.\ a static resource into a dynamical one \citep{Devetak-Winter,Resource-calculus}.
Therefore the need to go beyond conversion between static resources
is built into the very notion of entanglement as a resource. This
is supported by the fact that in physics everything, including a state,
can be viewed as a \emph{dynamical resource} \citep{Chiribella2008,Chiribella-purification,hardy2011}.
Therefore it is really \emph{necessary} to phrase entanglement theory
as a resource theory of \emph{quantum processes}. In these theories
the agent converts different dynamical resources by means of a restricted
set of\emph{ superchannels}. 

In light of this, in this article we expand the results originally
announced in Ref.~\citep{Dynamical-entanglement}, formulating a
rigorous treatment of the resource theory of entanglement as a resource
theory of processes (an independent work in this respect is Ref.~\citep{Wilde-entanglement}).

The generic resource will be a\emph{ bipartite channel} \citep{Shannon-bipartite,Bipartite}
rather than a bipartite state. A bipartite channel, represented in
Fig.~\ref{bipartitechannel}, is a channel with two inputs and two
outputs.
\begin{figure}
\begin{centering}
\includegraphics[width=1\columnwidth]{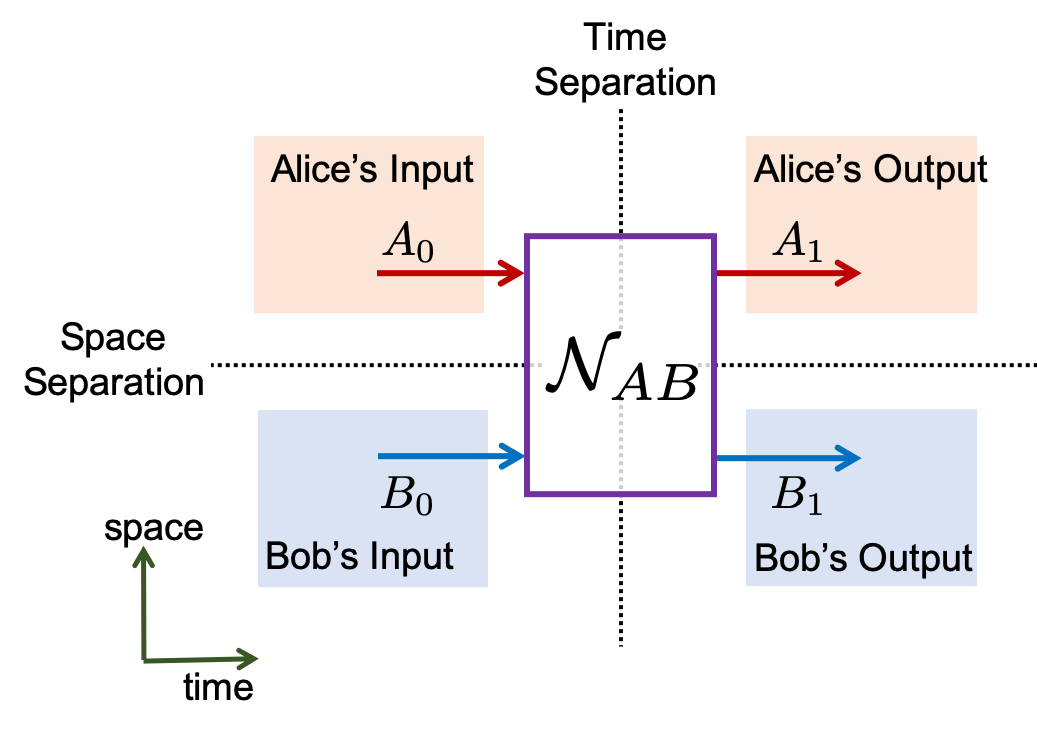}
\par\end{centering}
\caption{\label{bipartitechannel}The four regions of a bipartite channel.
Note the space separation between the two parties, Alice and Bob.
Unlike for bipartite states, we can also distinguish a temporal separation
between the input and the output of each party.}

\end{figure}
 We assume there is a spatial separation between the two inputs (and
also between the two outputs). This spatial separation is associated
with the presence of two space-like separated parties, Alice and Bob,
as for bipartite states. The novelty coming from the fact that we
are considering channels rather than states is that we also have a
time separation between the input side of the channel and its outputs.
This makes bipartite channels the most general resource for the study
of entanglement and, at the same time, the most versatile. Indeed,
if we trivialize (i.e.\ make 1-dimensional) the two inputs of a bipartite
channel, we recover the theory of entanglement for bipartite states.
On the other hand, if we consider classical outputs, we obtain the
``dual'' resource theory of entanglement for POVMs. We can also
consider other scenarios. For instance, if we trivialize Alice's output
and Bob's input, we get a 1-way channel from Alice to Bob, a situation
studied in Ref.~\citep{WW18}.

In this article, we apply the resource-theoretic constructions introduced
in Ref.~\citep{Gour-Scandolo-resource} to the resource theory of
entanglement for bipartite channels. In particular, we focus on PPT
and separable superchannels, for their greater mathematical simplicity,
in the same spirit as one considers PPT and separable channels to
study the entanglement of states. Our approach differs from Ref.~\citep{WW18}
in a twofold way. First, we study the most general resource: bipartite
channels, instead of just states and 1-way channels. This allows us
to generalize the notion of $\kappa$-entanglement \citep{WW18},
which we call \emph{max-logarithmic negativity}, in two distinct ways.
Second, we do not require PPT superchannels to have PPT pre- and post-processing
\citep{Leung}. This leads to a great simplification in the mathematical
treatment and the derivation of results, as all conditions on resource
conversion can be expressed in terms of semi-definite programs (SDPs).

We conclude the article by analyzing bound entanglement for bipartite
channels, showing that no entanglement can be distilled from PPT channels.
We also provide the example of a bound entangled POVM.

The article is organized as follows. After some background information
on superchannels and their Choi matrices presented in section~\ref{sec:Preliminaries},
the resource theory of entanglement for bipartite channels is introduced
in section~\ref{sec:entanglement}, where we define the basic resource-theoretic
protocols. In section~\ref{sec:PPT} we analyze the simplest resource
theory for entanglement from a mathematical point of view: the NPT
resource theory. We show that all resource conversion tasks can be
expressed in terms of SDPs, and, in particular, we provide an operational
interpretation for the max-logarithmic negativity. Separable superchannels
are introduced in section~\ref{sec:SEP}. We conclude the article
with a study of bound entanglement for bipartite channels in section~\ref{sec:Bound}.
Conclusions are drawn in section~\ref{sec:Conclusions-and-outlook}.

\section{Preliminaries\label{sec:Preliminaries}}

This section contains some basic notions that are extensively used
throughout this article. First we specify the notation we use, and
then we move to give a brief presentation of the main properties of
supermaps and superchannels. We conclude the section with an overview
of quantum combs.

\subsection{Notation}

Physical systems and their corresponding Hilbert spaces will be denoted
by $A$, $B$, $C$, etc, where the notation $AB$ means $A\otimes B$.
Dimensions will be denoted with vertical bars; so the dimension of
system $A$ will be denoted by $\left|A\right|$. The tilde symbol
will be reserved to indicate a replica of a system. For example, $\widetilde{A}$
denotes a replica of $A$, i.e.\ $\left|A\right|=\left|\widetilde{A}\right|$.
Density matrices acting on Hilbert spaces will be denoted by lowercase
Greek letters $\rho$, $\sigma$, etc, except for the maximally mixed
state (i.e.\ the uniform state), which will be denoted by $u_{A}:=\frac{1}{\left|A\right|}I_{A}$.

The set of all bounded operators acting on system $A$ is denoted
by $\mathfrak{B}\left(A\right)$, the set of all Hermitian matrices
acting on $A$ by $\mathrm{Herm}\left(A\right)$, and the set of all
density matrices acting on system $A$ by $\mathfrak{D}\left(A\right)$.
We use calligraphic letters $\mathcal{D}$, $\mathcal{E}$, $\mathcal{F}$,
etc.\ to denote quantum maps, reserving $\mathcal{T}$ to represent
the transpose map. The identity map on a system $A$ will be denoted
by $\mathsf{id}_{A}$. The set of all linear maps from $\mathfrak{B}\left(A\right)$
to $\mathfrak{B}\left(B\right)$ is denoted by $\mathfrak{L}\left(A\to B\right)$,
the set of all completely positive (CP) maps by $\mathrm{CP}\left(A\to B\right)$,
and the set of quantum channels, which are completely positive and
trace-preserving, by $\mathrm{CPTP}\left(A\to B\right)$. $\mathrm{Herm}\left(A\to B\right)$
will denote the real vector space of all Hermitian-preserving maps
in $\mathfrak{L}\left(A\to B\right)$. We will write $\mathcal{N}\geq0$
to mean that the map $\mathcal{N}\in\mathrm{Herm}\left(A\to B\right)$
is completely positive.

Unless otherwise specified, we will associate two subsystems $A_{0}$
and $A_{1}$ with every physical system $A$, referring, respectively,
to the input and output of the resource. Hence, any physical system
will be comprised of two subsystems $A=\left(A_{0},A_{1}\right)$,
even those representing a static resource, in which case we simply
have $\left|A_{0}\right|=1$. For simplicity, we will denote a channel
with a subscript $A$, e.g.\ $\mathcal{N}_{A}$, to mean that it
is an element of $\mathrm{CPTP}\left(A_{0}\to A_{1}\right)$. Similarly,
a bipartite channel in $\mathrm{CPTP}\left(A_{0}B_{0}\to A_{1}B_{1}\right)$
will be denoted by $\mathcal{N}_{AB}$. This notation makes the analogy
with bipartite states more transparent.

In this setting, when we consider $A=\left(A_{0},A_{1}\right)$, $B=\left(B_{0},B_{1}\right)$,
etc.\ comprised of input and output subsystems, the symbol $\mathfrak{L}\left(A\to B\right)$
refers to all linear maps from the vector space $\mathfrak{L}\left(A_{0}\to A_{1}\right)$
to the vector space $\mathfrak{L}\left(B_{0}\to B_{1}\right)$. Similarly,
$\mathrm{Herm}\left(A\to B\right)\subset\mathfrak{L}\left(A\to B\right)$
is a real vector space consisting of all the linear maps that take
elements in $\mathrm{Herm}\left(A_{0}\to A_{1}\right)$ to elements
in $\mathrm{Herm}\left(B_{0}\to B_{1}\right)$. In other terms, maps
in $\mathrm{Herm}\left(A\to B\right)$ take Hermitian-preserving maps
to Hermitian-preserving maps. Linear maps in $\mathfrak{L}\left(A\to B\right)$
and $\mathrm{Herm}\left(A\to B\right)$ will be called \emph{supermaps},
and will be denoted by capital Greek letters $\Theta$, $\Upsilon$,
$\Omega$, etc. In the following, to avoid confusion with the notation
for linear and Hermitian-preserving maps, whenever we mean linear
or Hermitian-preserving maps, the systems involved will have a subscript,
to make it clear that we are not considering pairs of systems. In
this setting, the identity supermap in $\mathfrak{L}\left(A\to A\right)$
will be denoted by $\mathbbm{1}_{A}$. 

We will use square brackets to denote the action of a supermap $\Theta_{A\to B}\in\mathfrak{L}\left(A\to B\right)$
on a linear map $\mathcal{N}_{A}\in\mathfrak{L}\left(A_{0}\to A_{1}\right)$.
For example, $\Theta_{A\to B}\left[\mathcal{N}_{A}\right]$ is a linear
map in $\mathfrak{L}\left(B_{0}\to B_{1}\right)$ obtained from the
action of the supermap $\Theta$ on the map $\mathcal{N}$. Moreover,
for a simpler notation, the identity supermap will not often appear
explicitly in equations; e.g.\ $\Theta_{A\to B}\left[\mathcal{N}_{RA}\right]$
will mean $\left(\mathbbm{1}_{R}\otimes\Theta_{A\to B}\right)\left[\mathcal{N}_{RA}\right]$.
Instead, the action of linear map (e.g.\ quantum channel) $\mathcal{N}_{A}\in\mathfrak{L}\left(A_{0}\to A_{1}\right)$
on a matrix $\rho\in\mathfrak{B}\left(A_{0}\right)$ will be written
with parentheses, i.e.\ $\mathcal{N}_{A}\left(\rho_{A_{0}}\right)\in\mathfrak{B}\left(A_{1}\right)$.

Finally, we adopt the following convention concerning partial traces:
when a system is missing, we take the partial trace over it. This
applies to matrices as well as to maps. For example, if $M_{AB}$
is a matrix on $A_{0}A_{1}B_{0}B_{1}$, $M_{AB_{0}}$ denotes the
partial trace on the missing system $B_{1}$: $M_{AB_{0}}:=\mathrm{Tr}_{B_{1}}\left[M_{AB}\right]$.

\subsection{Supermaps\label{subsec:Supermaps}}

In Refs.~\citep{Circuit-architecture,Hierarchy-combs,Gour2018} it
was shown that it is possible to construct the Choi matrix $\mathbf{J}_{AB}^{\Theta}$
of a quantum supermap $\Theta_{A\rightarrow B}$. In particular, we
can associate two linear maps with $\Theta_{A\rightarrow B}$ \citep{Gour2018}.
The first is the map $\mathcal{P}_{AB}^{\Theta}$, defined as
\[
\mathcal{P}_{AB}^{\Theta}:=\Theta_{\widetilde{A}\to B}\left[\Phi_{A\widetilde{A}}^{+}\right],
\]
where the map $\Phi_{A\widetilde{A}}^{+}$ acts on $\rho\in\mathfrak{B}\left(A_{0}\widetilde{A}_{0}\right)$
as 
\begin{equation}
\Phi_{A\widetilde{A}}^{+}\left(\rho_{A_{0}\widetilde{A}_{0}}\right)=\mathrm{Tr}\left[\rho_{A_{0}\widetilde{A}_{0}}\phi_{A_{0}\widetilde{A}_{0}}^{+}\right]\phi_{A_{1}\widetilde{A}_{1}}^{+},\label{eq:maxent}
\end{equation}
with $\phi_{A_{0}\widetilde{A}_{0}}^{+}:=\left|\phi^{+}\right\rangle \left\langle \phi^{+}\right|_{A_{0}\widetilde{A}_{0}}$
and $\left|\phi^{+}\right\rangle _{A_{0}\widetilde{A}_{0}}=\sum_{j}\left|jj\right\rangle _{A_{0}\widetilde{A}_{0}}$
is the unnormalized maximally entangled state (expressed in the Choi
basis). In other terms, the CP map $\Phi_{A\widetilde{A}}^{+}$ can
be viewed as a generalization of $\phi_{A_{0}\widetilde{A}_{0}}^{+}$.
With this construction, $\mathbf{J}_{AB}^{\Theta}$ can be defined
as the Choi matrix of the map $\mathcal{P}_{AB}^{\Theta}$.

The second representation of a supermap is in terms of a linear map
$\mathcal{Q}^{\Theta}:\mathfrak{B}\left(A_{1}B_{0}\right)\to\mathfrak{B}\left(A_{0}B_{1}\right)$,
which is defined as the map satisfying 
\[
\mathbf{J}_{AB}^{\Theta}:=\mathcal{Q}_{\widetilde{A}_{1}\widetilde{B}_{0}\to A_{0}B_{1}}^{\Theta}\left(\phi_{A_{1}\widetilde{A}_{1}}^{+}\otimes\phi_{B_{0}\widetilde{B}_{0}}^{+}\right),
\]
or as $\mathcal{Q}^{\Theta}:=\mathbbm{1}_{A}\otimes\Theta_{A\rightarrow B}\left[\mathcal{S}_{A}\right]$,
where $\mathcal{S}_{A}$ is the swap from $A_{1}$ to $A_{0}$. With
this second construction, it is apparent that the Choi matrix $\mathbf{J}_{AB}^{\Theta}$
of the supermap can be defined also as the Choi matrix of the the
map $\mathcal{Q}^{\Theta}$. On top of being useful for the definition
of the Choi matrix of a supermap, these two representations of a supermap,
$\mathcal{P}^{\Theta}$ and $\mathcal{Q}^{\Theta}$, will play a useful
role in the study of the entanglement of bipartite channels.

A \emph{superchannel} is a supermap $\Theta_{A\to B}\in\mathfrak{L}\left(A\to B\right)$
that takes quantum channels to quantum channels even when tensored
with the identity supermap \citep{Chiribella2008,Switch,Hierarchy-combs,Perinotti1,Perinotti2,Gour2018,Supermeasurements}.
More precisely, $\Theta_{A\to B}\in\mathfrak{L}\left(A\to B\right)$
is called a superchannel if it satisfies the following two conditions: 
\begin{enumerate}
\item For any trace-preserving map $\mathcal{N}_{A}\in\mathfrak{L}\left(A_{0}\to A_{1}\right)$,
the map $\Theta_{A\to B}\left[\mathcal{N}_{A}\right]$ is a trace-preserving
map in $\mathfrak{L}\left(B_{0}\to B_{1}\right)$.
\item For any system $R=\left(R_{0},R_{1}\right)$ and any bipartite CP
map $\mathcal{N}_{RA}\in\mathrm{CP}\left(R_{0}A_{0}\to R_{1}A_{1}\right)$,
the map $\Theta_{A\to B}\left[\mathcal{N}_{RA}\right]$ is also CP.
\end{enumerate}
We will also say that a supermap $\Theta_{A\to B}\in\mathfrak{L}\left(A\to B\right)$,
is \emph{completely positive} if it satisfies the second condition
above \citep{Chiribella2008,Gour2018}. Therefore, a superchannel
is a CP supermap that takes trace-preserving maps to trace-preserving
maps \citep{Gour2018,Supermeasurements}. We will denote the set of
superchannels from $A$ to $B$ by $\mathfrak{S}\left(A\rightarrow B\right)$.
Note that $\mathfrak{S}\left(A\rightarrow B\right)\subset\mathfrak{L}\left(A\rightarrow B\right)$.
In particular, for the Choi matrix of a superchannel, we have $\mathbf{J}_{A_{1}B_{0}}^{\Theta}=I_{A_{1}B_{0}}$
and $\mathbf{J}_{AB_{0}}^{\Theta}=\mathbf{J}_{A_{0}B_{0}}^{\Theta}\otimes u_{A_{1}}$.

The definitions seen so far are abstract; nevertheless, superchannels
are \emph{physical} objects that can be realized in terms of pre-
and post-processing that are both quantum channels \citep{Chiribella2008,Gour2018}.
Indeed, if $\Theta\in\mathfrak{S}\left(A\to B\right)$, then there
exist a Hilbert space $E$, with $\left|E\right|\leq\left|A_{0}B_{0}\right|$,
and two CPTP maps $\mathcal{F}\in\mathrm{CPTP}\left(B_{0}\to EA_{0}\right)$
and $\mathcal{E}\in\mathrm{CPTP}\left(EA_{1}\to B_{1}\right)$ such
that, for all $\mathcal{N}_{A}\in\mathfrak{L}\left(A_{0}\to A_{1}\right)$,
\[
\Theta\left[\mathcal{N}_{A}\right]=\mathcal{E}_{EA_{1}\to B_{1}}\circ\mathcal{N}_{A_{0}\to A_{1}}\circ\mathcal{F}_{B_{0}\to EA_{0}}.
\]

\subsection{Quantum combs}

Quantum combs are multipartite channels with a well-defined causal
structure \citep{Gutoski,Circuit-architecture,Hierarchy-combs,Gutoski2,Chiribella2016,Gutoski3}.
They generalize the notion of superchannels to objects that take several
channels as input, and output a channel (see Refs.~\citep{Circuit-architecture,Hierarchy-combs}
for more details, and a further generalization where the input and
the output of combs are combs themselves). We will denote a comb with
$n$ channel slots as input by $\mathscr{C}_{n}$, and its action
on $n$ channels by $\mathscr{C}_{n}\left[\mathcal{N}_{1},\dots,\mathcal{N}_{n}\right]$.
The causal relation between the different slots ensures that each
such comb can be realized with $n+1$ channels $\mathcal{E}_{1},\dots,\mathcal{E}_{n+1}$.
We therefore associate a quantum channel 
\[
\mathcal{Q}^{\mathscr{C}_{n}}:=\mathcal{E}_{n+1}\circ\mathcal{E}_{n}\circ\dots\circ\mathcal{E}_{1}
\]
with every comb. Note that the quantum channel $\mathcal{Q}^{\mathscr{C}_{n}}$
has a causal structure in the sense that the input to $\mathcal{E}_{k}$
cannot affect the output of $\mathcal{E}_{k-1}$ for any $k=2,\dots,n+1$.

The Choi matrix of the comb is defined as the Choi matrix of $\mathcal{Q}^{\mathscr{C}_{n}}$.
Owing to the causal structure of $\mathcal{Q}^{\mathscr{C}_{n}}$,
the marginals of the Choi matrix of $\mathscr{C}_{n}$ satisfy similar
relations to the marginals of the Choi matrix of a superchannel (see
Refs.~\citep{Circuit-architecture,Hierarchy-combs} for more details).

\section{Dynamical entanglement theory\label{sec:entanglement}}

Recall that with one ebit, thanks to quantum teleportation \citep{Teleportation},
we can simulate a qubit channel from Alice to Bob using LOCC \citep{LOCC1,LOCC2,Lo-Popescu},
and vice versa \citep{Devetak-Winter,Resource-calculus}. Therefore
one ebit (a \emph{static} resource) is equivalent to a dynamical one:
a qubit channel. Considering \emph{bipartite channels} \citep{Bipartite}
in $\mathrm{CPTP}\left(A_{0}B_{0}\rightarrow A_{1}B_{1}\right)$ (see
Fig.~\ref{bipartitechannel}), we can understand the qubit identity
channel from $A_{0}$ to $B_{1}$ as the maximal resource under LOCC
as long as $\left|A_{1}\right|=\left|B_{0}\right|=1$. It is maximal
because, by using it, every other channel can be implemented between
$A_{0}$ and $B_{1}$.

Now let us generalize this situation by analyzing what the maximal
resource is when all systems are non-trivial, and specifically $\left|A_{0}\right|=\left|A_{1}\right|=\left|B_{0}\right|=\left|B_{1}\right|=d$.
In Fig.~\ref{swap} we show that the \noun{swap} operation is a maximal
resource.
\begin{figure}
\begin{centering}
\includegraphics[width=1\columnwidth]{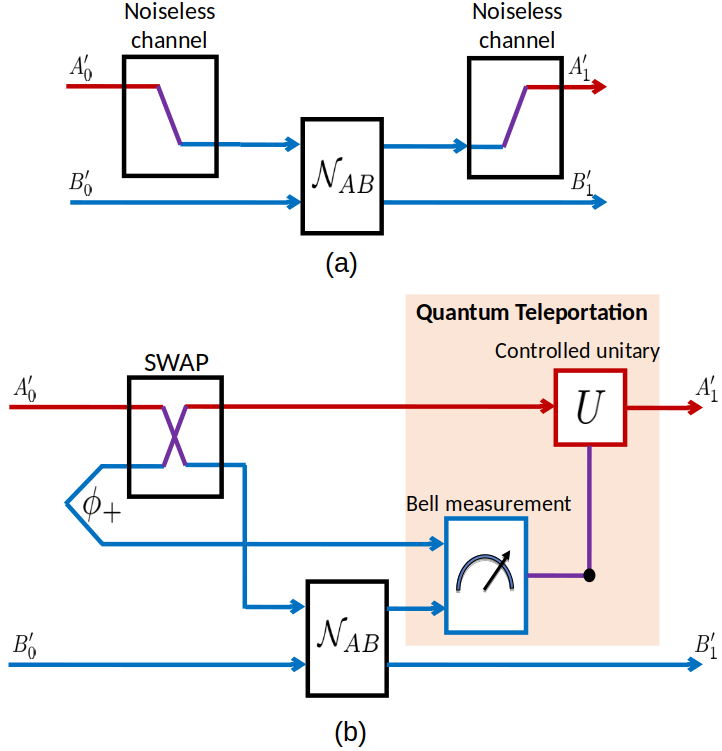}
\par\end{centering}
\caption{\label{swap}(a) Simulation of an arbitrary channel $\mathcal{N}_{A'B'}$
with two noiseless channels. (b) Simulation of an arbitrary channel
$\mathcal{N}_{A'B'}$ with the \noun{swap} resource and 1-way LOCC.}

\end{figure}
 Note that the \noun{swap} operator can produce 2 e-dits, and can
also be simulated by 2 e-dits. Therefore, the entanglement of the
\noun{swap} operator is 2 e-dits. Note also that the \noun{swap} operator
is the maximal resource even if the set of free operations allows
only 1-way classical communication. On the other hand, in the quantum
resource theory in which free operations consists of only local operations
and shared entanglement (LOSE) \citep{LOSE}, but no classical communication,
then two noiseless channels, one from $A_{0}\to B_{1}$ and one from
$B_{0}\to A_{1}$, are more resourceful than the \noun{swap} operator.
This is because the \noun{swap} operator is restricted to act \emph{simultaneously}
on both input systems. This example demonstrates that in general,
two channels $\mathcal{N}_{A_{0}\to B_{1}}$ and $\mathcal{M}_{B_{0}\to A_{1}}$
can be more resourceful than their tensor product $\mathcal{N}_{A_{0}\to B_{1}}\otimes\mathcal{M}_{B_{0}\to A_{1}}$
since they can be used \emph{at different times}.

The fact that a tuple of $n$ channels can be a greater resource than
their tensor product was also discussed in Ref.~\citep{Resource-channels-2}
(cf.\ also Ref.~\citep{Gour-Scandolo-resource}). In the following,
however, we will focus mainly on a \emph{single} resource at a time,
in this case a single bipartite channel.

\subsection{Simulation of channels: cost and distillation}

Following Refs.~\citep{Resource-theories,Resource-channels-1,Resource-channels-2,Gour-Scandolo-resource},
in Fig.~\ref{LOCC} we illustrate the most general LOCC superchannel
that can act on a bipartite channel.
\begin{figure}
\begin{centering}
\includegraphics[width=1\columnwidth]{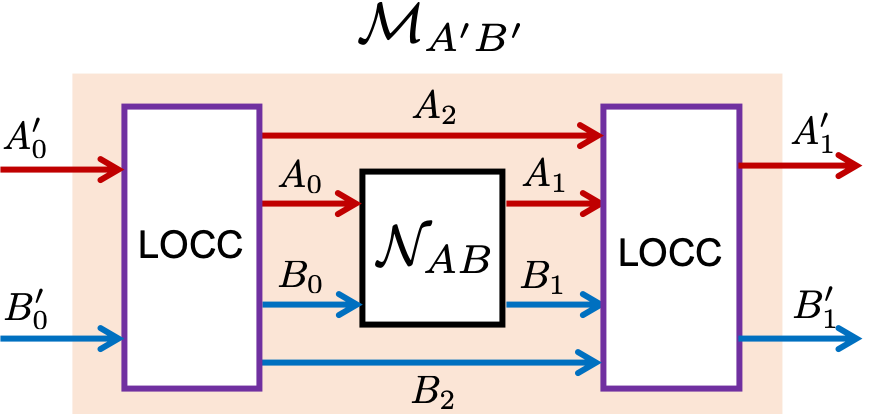}
\par\end{centering}
\caption{\label{LOCC}The action of an LOCC superchannel on one copy of a bipartite
channel $\mathcal{N}_{AB}$. The resulting channel is $\mathcal{M}_{A'B'}$.
Note that this superchannel uses the dynamical resource $\mathcal{N}_{AB}$
to simulate another channel $\mathcal{M}_{A'B'}$.}

\end{figure}
 The superchannel consists of a pre-processing and a post-processing
channel that are both LOCC. Moreover, the side channel, corresponding
to the memory in the realization of a superchannel, consists of two
parts: $A_{2}$ on Alice's side and $B_{2}$ on Bob's side. We denote
the set of such superchannels by $\mathrm{LOCC}\left(AB\to A'B'\right)$. 

The discussion at the beginning of section~\ref{sec:entanglement}
shows that ebits remain the units to quantify the entanglement of
a bipartite channel. Indeed, two ebits can be used to simulate any
bipartite channel in which the two input and two output systems are
all qubits. Therefore, even in the resource theory of entanglement
of bipartite channels one can define operational tasks in a very similar
fashion to the state domain. For example, in Figs.~\ref{nLOCC}(a)
and \ref{nLOCC}(b) we illustrate parallel \citep{Berta-cost} and
adaptive strategies \citep{Pirandola-LOCC,Kaur2017,Wilde-cost,Wilde-entanglement,Pirandola-network}
to distill static entanglement out of a dynamical resource.
\begin{figure}
\begin{centering}
\includegraphics[width=1\columnwidth]{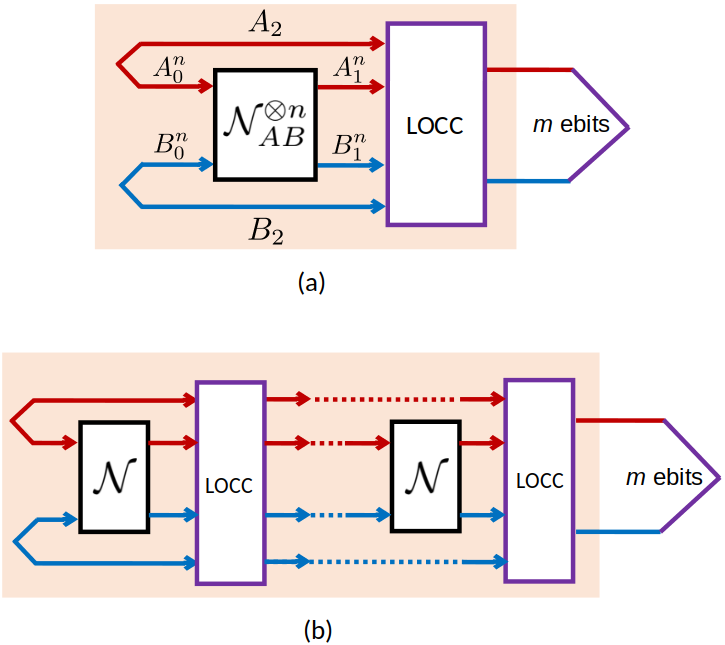}
\par\end{centering}
\caption{\label{nLOCC}The action of an LOCC superchannel on $n$ copies of
the bipartite channel $\mathcal{N}_{AB}$. (a) Parallel strategy for
the distillation of $m$ ebits out of $n$ copies of $\mathcal{N}_{AB}$.
(b) Adaptive strategy for the distillation of $m$ ebits out of $n$
subsequent uses of $\mathcal{N}_{AB}$.}

\end{figure}
 Since the parallel scheme is a special instance of the adaptive strategy,
the distillable entanglement cannot be smaller when using the adaptive
scheme. However, in section~\ref{sec:Bound} we will see that there
are bipartite entangled channels from which \emph{no} distillation
is possible, no matter what strategy is applied. This generalizes
the notion of bound entanglement \citep{Bound-entanglement} to bipartite
channels.

Similar to distillation, also the entanglement cost of a bipartite
channel can be divided into two types: parallel and adaptive. In the
parallel scheme, the goal is to simulate $\mathcal{N}_{AB}^{\otimes n}$,
i.e.\ $n$ copies of $\mathcal{N}_{AB}$ all acting simultaneously
(see Fig.~\ref{cost}a). On the other hand, the goal of the adaptive
scheme is to simulate $n$ copies of $\mathcal{N}_{AB}$ in a time
sequential order (see Fig.~\ref{cost}b).
\begin{figure}
\begin{centering}
\includegraphics[width=1\columnwidth]{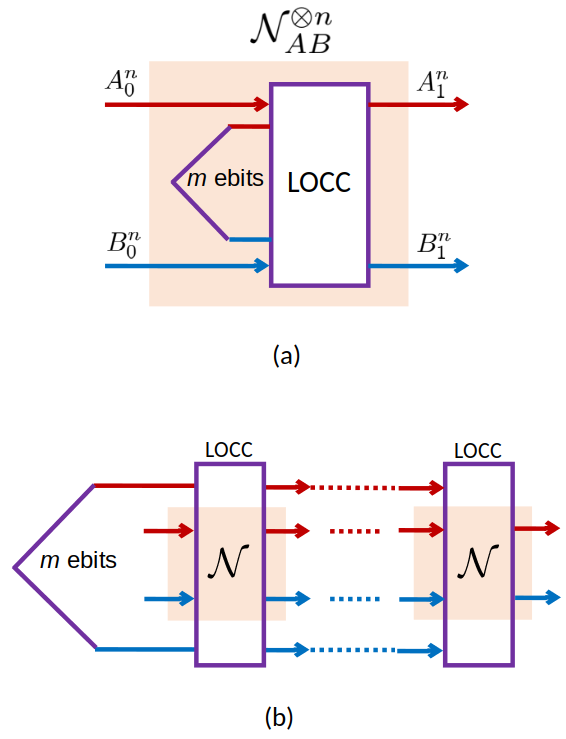}
\par\end{centering}
\caption{\label{cost}The cost of simulating a bipartite channel. (a) Parallel
strategy: consumption of $m$ ebits to simulate $\mathcal{N}_{AB}^{\otimes n}$.
(b) Adaptive strategy: consumption $m$ ebits to simulate $n$ subsequent
uses of $\mathcal{N}_{AB}$.}

\end{figure}
 Both schemes use ebits to simulate the channels. For the same reason
as for the distillation case, note that the cost of simulating $n$
sequentially ordered channels cannot be greater than the cost in the
parallel scheme. Owing to the complexity of the adaptive scheme, in
this paper we will focus mostly on the parallel one.

Now we are ready to give the formal definitions of entanglement costs
and distillable entanglement of bipartite channels. First of all,
note that in entanglement theory, the conversion distance for any
two channels $\mathcal{N}_{AB}$ and $\mathcal{M}_{A'B'}$, introduced
in Ref.~\citep{Gour-Scandolo-resource} and inspired by Ref.~\citep{Tomamichel},
is given by
\begin{align*}
 & d_{\mathrm{LOCC}}\left(\mathcal{N}_{AB}\to\mathcal{M}_{A'B'}\right)\\
 & =\frac{1}{2}\inf_{\Theta\in\mathrm{LOCC}\left(AB\to A'B'\right)}\left\Vert \Theta_{AB\to A'B'}\left[\mathcal{N}_{AB}\right]-\mathcal{M}_{A'B'}\right\Vert _{\diamond},
\end{align*}
where the optimization is over the set of LOCC superchannels. Typically,
the computation of this quantity is NP-hard. To see why, consider
the special case in which $\mathcal{N}_{AB}$ is a bipartite separable
\emph{state} (i.e.\ $\left|A_{0}\right|=\left|B_{0}\right|=1$),
and $\mathcal{M}_{AB}$ is some (possibly entangled) bipartite state
as well. In this case, the computation of the conversion distance
would determine if the bipartite state $\mathcal{M}_{AB}$ is entangled
or not, but this is known to be NP-hard \citep{NP-hard1,NP-hard2}.

Furthermore, we know that if $\Theta\in\mathrm{LOCC}\left(AB\to A'B'\right)$
then the bipartite channel $\mathcal{Q}_{AB\to A'B'}^{\Theta}$ is
also LOCC, while the condition that $\mathcal{Q}_{AB\to A'B'}^{\Theta}$
is LOCC is most likely insufficient to ensure that $\Theta\in\mathrm{LOCC}\left(AB\to A'B'\right)$.
This adds another layer of complexity to the problem of computing
$d_{\mathrm{LOCC}}$. In section~\ref{sec:PPT} we will see that
this additional complexity persists even when considering simpler
sets of operations, like PPT channels \citep{PPT1,PPT2}.

Since in entanglement theory there exists a unique (up to local unitaries)
maximal static resource, the single-shot entanglement cost and entanglement
distillation with error $\varepsilon\geq0$ are given respectively
by
\begin{align*}
 & \mathrm{COST}_{\mathrm{LOCC},\varepsilon}^{\left(1\right)}\left(\mathcal{N}_{AB}\right)\\
 & :=\log_{2}\min_{m\in\mathbb{Z}_{+}}\left\{ m:d_{\mathrm{LOCC}}\left(\phi_{m}^{+}\to\mathcal{N}_{AB}\right)\leq\varepsilon\right\} 
\end{align*}
and
\begin{align*}
 & \mathrm{DISTILL}_{\mathrm{LOCC},\varepsilon}^{\left(1\right)}\left(\mathcal{N}_{AB}\right)\\
 & :=\log_{2}\max_{m\in\mathbb{Z}_{+}}\left\{ m:d_{\mathrm{LOCC}}\left(\mathcal{N}_{AB}\to\phi_{m}^{+}\right)\leq\varepsilon\right\} ,
\end{align*}
where $\phi_{m}^{+}$ is a (normalized) maximally entangled state
with Schmidt rank $m$. Here the optimizations are over this Schmidt
rank $m$. Then the entanglement cost and the distillable entanglement
of a dynamical resource in the asymptotic regime are defined respectively
as
\[
\mathrm{COST}_{\mathrm{LOCC}}\left(\mathcal{N}_{AB}\right):=\lim_{\varepsilon\to0^{+}}\liminf_{n}\frac{1}{n}\mathrm{COST}_{\mathrm{LOCC},\varepsilon}^{\left(1\right)}\left(\mathcal{N}_{AB}^{\otimes n}\right)
\]
and
\begin{align*}
 & \mathrm{DISTILL}_{\mathrm{LOCC}}\left(\mathcal{N}_{AB}\right)\\
 & :=\lim_{\varepsilon\to0^{+}}\limsup_{n}\frac{1}{n}\mathrm{DISTILL}_{\mathrm{LOCC},\varepsilon}^{\left(1\right)}\left(\mathcal{N}_{AB}^{\otimes n}\right).
\end{align*}
These definitions assume the parallel scheme. In the adaptive scheme,
the entanglement cost and the distillable entanglement are defined
accordingly, as per Ref.~\citep{Gour-Scandolo-resource}.

\subsection{Measures of dynamical entanglement}

In this section we discuss a few measures that quantify the entanglement
of a bipartite channel. We also examine the form that the complete
family of resource measures introduced in Ref.~\citep{Gour-Scandolo-resource}
takes in entanglement theory.

A function $E:\mathrm{CPTP}\left(A_{0}B_{0}\to A_{1}B_{1}\right)\to\mathbb{R}$
is called a \emph{measure of dynamical entanglement} if it does not
increase under LOCC superchannels. It is called \emph{dynamical entanglement
monotone} if it is convex, and does not increase on average under
LOCC superinstruments \citep{Supermeasurements}. Some measures of
dynamical resources are discussed in Refs.~\citep{Pirandola-LOCC,Das-PhD,Resource-channels-1,Resource-channels-2,Gour-Winter,Fang-Fawzi,Wilde-entanglement,Dynamical-entanglement,Gour2018a,Gour-Scandolo-resource,Gour-divergences,Jencova,Gour-Sarah}.
Specifically, for bipartite entanglement the \emph{relative entropy
of dynamical entanglement} can be defined as 
\[
E_{\text{rel}}\left(\mathcal{N}_{AB}\right)=\inf_{\mathcal{M}\in\mathrm{LOCC}}D\left(\mathcal{N}_{AB}\middle\|\mathcal{M}_{AB}\right).
\]
Note that we are using the infimum rather than the minimum because
the set of LOCC channels is not topologically closed \citep{Chitambar2014}.

Moreover, any measure of static entanglement $E$ that is monotonic
under separable channels (in particular, under LOCC) can be extended
to bipartite channels in two different ways \citep{Resource-channels-1,Resource-channels-2,Gour-Winter,Gour-Scandolo-resource}.
In the first, we consider the amortized extension (cf.\ also Refs.~\citep{Das-PhD,Das-secret})
\begin{widetext}
\[
E^{\left(1\right)}\left(\mathcal{N}_{AB}\right):=\sup_{\sigma\in\mathfrak{D}\left(A_{0}'B_{0}'A_{0}B_{0}\right)}\left\{ E\left(\mathcal{N}_{A_{0}B_{0}\to A_{1}B_{1}}\left(\sigma_{A_{0}'B_{0}'A_{0}B_{0}}\right)\right)-E\left(\sigma_{A_{0}'B_{0}'A_{0}B_{0}}\right)\right\} ,
\]
\end{widetext}

\noindent where $A_{0}'$ and $B_{0}'$ are additional reference systems
in Alice's and Bob's sides, respectively, and the optimization is
over all density matrices on the system $A_{0}'B_{0}'A_{0}B_{0}$.
The other extension is given by
\begin{align*}
 & E^{\left(2\right)}\left(\mathcal{N}_{AB}\right)\\
 & :=\sup_{\sigma\in\mathrm{SEP}\left(A_{0}'A_{0}:B_{0}'B_{0}\right)}E\left(\mathcal{N}_{A_{0}B_{0}\to A_{1}B_{1}}\left(\sigma_{A_{0}'B_{0}'A_{0}B_{0}}\right)\right),
\end{align*}
where $\mathrm{SEP}\left(A_{0}'A_{0}:B_{0}'B_{0}\right)$ denotes
the set of separable states between Alice and Bob. Both of the above
extensions of $E$ can be proved to be non-increasing under separable
superchannels \citep{Gour-Winter}.

Now we introduce the complete family of dynamical entanglement measures,
following our construction in Ref.~\citep{Gour-Scandolo-resource}.
For any (fixed) bipartite channel $\mathcal{P}\in\mathrm{CPTP}\left(A_{0}'B_{0}'\to A_{1}'B_{1}'\right)$,
define (see Ref.~\citep{Gour-Scandolo-resource})
\[
E_{\mathcal{P}}\left(\mathcal{N}_{AB}\right):=\sup_{\Theta\in\mathrm{LOCC}\left(AB\to A'B'\right)}\mathrm{Tr}\left[J_{A'B'}^{\mathcal{P}}J_{A'B'}^{\Theta\left[\mathcal{N}\right]}\right],
\]
where $\mathcal{N}_{A}\in\mathrm{CPTP}\left(A_{0}B_{0}\to A_{1}B_{1}\right)$,
and $J$ is the Choi matrix of the channel in its superscript. Note
again that we are using the supremum instead of the maximum because
the set of LOCC channels is not closed. This function may not vanish
on LOCC channels; if we want so, we need to subtract $\sup_{\mathcal{M}\in\mathrm{LOCC}\left(A'B'\right)}\mathrm{Tr}\left[J_{A'B'}^{\mathcal{P}}J_{A'B'}^{\mathcal{M}}\right]$.
As explained in Ref.~\citep{Gour-Scandolo-resource}, this defines
a new non-negative measure of dynamical entanglement, which vanishes
on LOCC channels. Furthermore, the set of functions $\left\lbrace E_{\mathcal{P}}\right\rbrace $
is complete, in the sense that a bipartite channel $\mathcal{N}_{AB}$
can be converted within the topological closure of LOCC superchannels
into another bipartite channel $\mathcal{E}_{A'B'}$ if and only if
\begin{equation}
E_{\mathcal{P}}\left(\mathcal{N}_{AB}\right)\geq E_{\mathcal{P}}\left(\mathcal{E}_{A'B'}\right)\label{eq:formmonotone}
\end{equation}
for every $\mathcal{P}\in\mathrm{CPTP}\left(A_{0}'B_{0}'\to A_{1}'B_{1}'\right)$.

A natural question to ask is whether it is possible to find another
family of measures of dynamical entanglement that is finite, but at
the same time complete. However, in Ref.~\citep{Gour-infinite} it
was proved that any such complete family of entanglement measures
\emph{must} be infinite. Nevertheless, our family $\left\{ E_{\mathcal{P}}\right\} $
can be made \emph{countable} since we can remove from it all the channels
$\mathcal{P}$ whose Choi matrix includes coefficients that are irrational.
This can be done because, by construction, each function $E_{\mathcal{P}}$
is continuous in $\mathcal{P}$. Since the set of all channels $\mathcal{P}$
whose Choi matrices involve only rational coefficients is dense in
the set of all Choi matrices, by continuity it follows that, if Eq.~\eqref{eq:formmonotone}
holds for all such rational $\mathcal{P}$s, it holds also for all
$\mathcal{P}\in\mathrm{CPTP}\left(A_{0}'B_{0}'\to A_{1}'B_{1}'\right)$.
We conclude that our family $\left\{ E_{\mathcal{P}}\right\} $ is
optimal, in the sense that there is no other complete family of measures
of dynamical entanglement that characterizes the LOCC entanglement
of a bipartite channel more efficiently.

Despite the various interesting properties of the measures of dynamical
entanglement discussed in this section, they are all extremely hard
to compute due to the complexity of LOCC channels and superchannels.
We leave the discussion of more computationally manageable measures
to section~\ref{subsec:NPT-entanglement-measures}.

\subsection{Entanglement of bipartite POVMs\label{subsec:Entanglement-POVM}}

We end this section on the general properties of the resource theory
of dynamical entanglement with a short discussion on entanglement
of bipartite POVMs. A bipartite channel $\mathcal{N}\in\mathrm{CPTP}\left(A_{0}B_{0}\to A_{1}B_{1}\right)$
for which the output system $A_{1}B_{1}$ is classical can be viewed
as a POVM. In this case, the channel can be expressed as
\[
\mathcal{N}_{AB}\left(\rho_{A_{0}B_{0}}\right)=\sum_{x,y}\mathrm{Tr}\left[\rho_{A_{0}B_{0}}E_{A_{0}B_{0}}^{xy}\right]\left|xy\right\rangle \left\langle xy\right|_{A_{1}B_{1}},
\]
where the set of matrices $\left\{ E_{A_{0}B_{0}}^{xy}\right\} _{x,y}$
forms a POVM, and $\left\{ \left|xy\right\rangle \right\} _{x,y}$
is an orthonormal basis of $A_{1}B_{1}$. Such channels are fully
characterized by the condition $\mathcal{D}_{A_{1}B_{1}}\circ\mathcal{N}_{AB}=\mathcal{N}_{AB}$,
where $\mathcal{D}_{A_{1}B_{1}}$ is the completely dephasing channel
on system $A_{1}B_{1}$ (with respect to the fixed classical basis).
Note that $\mathcal{D}_{A_{1}B_{1}}\in\mathrm{LOCC}\left(A_{1}B_{1}\to A_{1}B_{1}\right)$.
\begin{lem}
Let $\mathcal{N}\in\mathrm{CPTP}\left(A_{0}B_{0}\to A_{1}B_{1}\right)$
be a bipartite POVM. Then 
\begin{equation}
E_{\mathrm{rel}}\left(\mathcal{N}_{AB}\right)=\inf_{\substack{\mathcal{M}\in\mathrm{LOCC}\\
\mathcal{D}_{A_{1}B_{1}}\circ\mathcal{M}_{AB}=\mathcal{M}_{AB}
}
}D\left(\mathcal{N}_{AB}\middle\|\mathcal{M}_{AB}\right).\label{aa}
\end{equation}
\end{lem}

\begin{proof}
Clearly, by definition $E_{\mathrm{rel}}\left(\mathcal{N}_{AB}\right)$
is less than or equal to the right-hand side of Eq.~\eqref{aa}.
Let us prove the converse inequality. We have
\begin{align*}
E_{\mathrm{rel}}\left(\mathcal{N}_{AB}\right) & =\inf_{\mathcal{M}\in\mathrm{LOCC}}D\left(\mathcal{N}_{AB}\middle\|\mathcal{M}_{AB}\right)\\
 & \geq\inf_{\mathcal{M}\in\mathrm{LOCC}}D\left(\mathcal{D}_{A_{1}B_{1}}\circ\mathcal{N}_{AB}\middle\|\mathcal{D}_{A_{1}B_{1}}\circ\mathcal{M}_{AB}\right),
\end{align*}
where the inequality follows from the generalized data-processing
inequality \citep{Gour-Winter}. Now recall that, being a POVM, $\mathcal{D}_{A_{1}B_{1}}\circ\mathcal{N}_{AB}=\mathcal{N}_{AB}$.
Therefore $E\left(\mathcal{N}_{AB}\right)\geq\inf_{\mathcal{M}\in\mathrm{LOCC}}D\left(\mathcal{N}_{AB}\middle\|\mathcal{D}_{A_{1}B_{1}}\circ\mathcal{M}_{AB}\right)$.
Hence we conclude that 
\[
E_{\mathrm{rel}}\left(\mathcal{N}_{AB}\right)=\inf_{\substack{\mathcal{M}\in\mathrm{LOCC}\\
\mathcal{D}_{A_{1}B_{1}}\circ\mathcal{M}_{AB}=\mathcal{M}_{AB}
}
}D\left(\mathcal{N}_{AB}\middle\|\mathcal{M}_{AB}\right).
\]
\end{proof}
The above lemma demonstrates that the relative entropy of entanglement
of a bipartite POVM can be viewed as its relative entropy distance
to the set of LOCC POVMs (rather than arbitrary bipartite LOCC channels).

Now, note that if systems $A_{1}$ and $B_{1}$ are classical, we
can view them as a single classical system (since classical communication
is free), and instead of using two indices $x,y$ to characterize
the POVM, it makes more sense to use just a single index, say $x$.
In this setting, the above lemma can be used to calculate the relative
entropy of process-entanglement for a POVM $\left\{ N_{A_{0}B_{0}}^{x}\right\} $.
Consider the associated quantum-to-classical channel $\mathcal{N}_{A_{0}B_{0}\rightarrow X}\left(\rho_{A_{0}B_{0}}\right)=\sum_{x=1}^{\left|X\right|}\mathrm{Tr}\left[\rho_{A_{0}B_{0}}N_{A_{0}B_{0}}^{x}\right]\left|x\right\rangle \left\langle x\right|_{X}$,
and an LOCC POVM $\left\{ F_{A_{0}B_{0}}^{y}\right\} $, with its
associated quantum-to-classical channel $\mathcal{F}_{A_{0}B_{0}\rightarrow Y}\left(\rho_{A_{0}B_{0}}\right)=\sum_{y=1}^{\left|Y\right|}\mathrm{Tr}\left[\rho_{A_{0}B_{0}}F_{A_{0}B_{0}}^{y}\right]\left|y\right\rangle \left\langle y\right|_{Y}$.
Now, possibly by completing one of the two POVMs with some zero elements,
we can always take $X=Y$. To calculate the channel divergence we
have to evaluate $\mathcal{N}_{A_{0}B_{0}\rightarrow X}$ and $\mathcal{F}_{A_{0}B_{0}\rightarrow X}$
on any pure state $\psi_{RA_{0}B_{0}}$, where $R$ is isomorphic
to $A_{0}B_{0}$ \citep{Cooney2016,Datta}. Recall that $\psi_{RA_{0}B_{0}}=\left(I_{R}\otimes\sqrt{\gamma_{A_{0}B_{0}}}U_{A_{0}B_{0}}\right)\phi_{RA_{0}B_{0}}^{+}\left(I_{R}\otimes U_{A_{0}B_{0}}^{\dagger}\sqrt{\gamma_{A_{0}B_{0}}}\right)$,
where $\gamma_{A}\in\mathfrak{D}\left(A_{0}B_{0}\right)$ and $U_{A_{0}B_{0}}$
is some unitary. After some calculations, we obtain
\begin{widetext}
\[
E\left(\left\{ N^{x}\right\} \right)=\inf_{\left\{ F_{x}\right\} \in\mathrm{LOCC}}\max_{\gamma,U}D\left(\sum_{x}U\sqrt{\gamma}\left(N^{x}\right)^{T}\sqrt{\gamma}U^{\dagger}\otimes\left|x\right\rangle \left\langle x\right|\middle\|\sum_{x}U\sqrt{\gamma}(F^{x})^{T}\sqrt{\gamma}U^{\dagger}\otimes\left|x\right\rangle \left\langle x\right|\right).
\]
\end{widetext}

\noindent By the properties of $D$, we have finally
\begin{align*}
 & E\left(\left\{ N^{x}\right\} \right)\\
 & =\inf_{\left\{ F_{x}\right\} \in\mathrm{LOCC}}\max_{\gamma}\sum_{x}D\left(\sqrt{\gamma}\left(N^{x}\right)^{T}\sqrt{\gamma}\middle\|\sqrt{\gamma}\left(F^{x}\right)^{T}\sqrt{\gamma}\right).
\end{align*}

Using the protocol of entanglement swapping~\citep{Entanglement-swapping},
we can use the entanglement of POVMs to produce static entanglement.
This is illustrated in Fig.~\ref{loccpovmsuperchannel}.
\begin{figure}
\begin{centering}
\includegraphics[width=1\columnwidth]{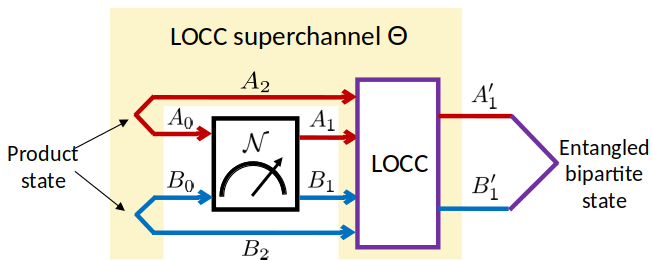}
\par\end{centering}
\caption{\label{loccpovmsuperchannel}Here $A_{1}$ and $B_{1}$ are classical
systems. The action of an LOCC superchannel $\Theta$ on a bipartite
channel with classical output can produce an entangled state.}

\end{figure}

\section{NPT entanglement of a bipartite channel\label{sec:PPT}}

Entanglement theory is hard to study due to the complexity of LOCC
channels \citep{Owari-2008a,Nathanson-2013a,Chitambar2014,Chitambar-2017a,Wakakuwa-2016a}
and the fact that even determining whether a given state is entangled
or not is known to be NP-hard \citep{NP-hard1,NP-hard2}. For this
reason, much of the work in recent years on entanglement theory involved
the replacement of LOCC with a larger set of free operations that
are more computationally friendly (see e.g.\ Ref.~\citep{Chitambar-2018b}
and references therein). One such set is the set of \emph{separable
operations} (or in short SEP; cf.\ section~\ref{sec:SEP}) \citep{SEP,Rains-SEP,PPT1},
another one is the set of \emph{PPT operations} \citep{PPT1,PPT2}.
Both sets are larger than LOCC, but the set of PPT operations is much
larger than both LOCC and SEP operations, as it contains, for instance,
PPT bound entangled states \citep{PHorodecki,Bound-entanglement}
(viewed as PPT channels with trivial input). Yet, among them, the
set of PPT operations has the simplest characterization, and can be
used to provide insights into LOCC entanglement, including various
bounds on LOCC tasks.

Bipartite states with positive (semi-definite) partial transpose (called
PPT states) were first discussed in Refs.~\citep{PPT-Peres,PPT-Horodecki}
in the context of entanglement theory. A few years later Rains \citep{PPT1,PPT2}
defined PPT bipartite channels for the first time (of which LOCC or
SEP channels are a special type), and used them to find an upper bound
on the distillable entanglement. In this section we consider PPT superchannels
\citep{Leung}, and use them for the study of entanglement of bipartite
channels. We will see that several of the optimization problems introduced
in the previous sections can be solved with SDPs in this theory of
entanglement, called the theory of \emph{NPT entanglement}. We start
with a few notations that will be very useful in the following.

Denote the transpose supermap by $\Upsilon_{B}\in\mathfrak{L}\left(B\to B\right)$:
\[
\Upsilon_{B}\left[\mathcal{N}_{B}\right]:=\mathcal{T}_{B_{1}}\circ\mathcal{N}_{B_{0}\to B_{1}}\circ\mathcal{T}_{B_{0}},
\]
for all $\mathcal{N}_{B}\in\mathfrak{L}\left(B_{0}\to B_{1}\right)$,
where $\mathcal{T}_{B_{0}}$ and $\mathcal{T}_{B_{1}}$ are the transpose
maps on the input and output systems, respectively. In Refs.~\citep{PPT1,PPT2}
the symbol $\Gamma$ was used to indicate the partial transpose of
a bipartite channel; that is, 
\[
\mathcal{N}_{AB}^{\Gamma}:=\left(\mathbbm{1}_{A}\otimes\Upsilon_{B}\right)\left[\mathcal{N}_{AB}\right].
\]
In the following we adopt the convention to always choose Bob's systems
(i.e.\ those denoted by $B$) to apply the (partial) transpose to.
With these notations, the set of PPT maps in $\mathrm{CP}\left(A_{0}B_{0}\to A_{1}B_{1}\right)$
is defined as
\begin{align*}
 & \mathrm{PPT}\left(A_{0}B_{0}\to A_{1}B_{1}\right)\\
 & =\left\{ \mathcal{N}\in\mathrm{CP}\left(A_{0}B_{0}\to A_{1}B_{1}\right):\mathcal{N}_{AB}^{\Gamma}\geq0\right\} .
\end{align*}
Note that PPT maps are defined as general CP maps, not necessarily
as channels. PPT maps have several useful properties. First, $\mathcal{N}_{AB}\in\mathrm{PPT}\left(A_{0}B_{0}\to A_{1}B_{1}\right)$
if and only if its Choi matrix $J_{AB}^{\mathcal{N}}$ satisfies $J_{AB}^{\mathcal{N}}\geq0$
and $\left(J_{AB}^{\mathcal{N}}\right)^{T_{B}}\geq0$. The former
condition implies that $\mathcal{N}_{AB}$ is a CP map, and the latter
ensures that it is PPT. The latter follows from the identity 
\begin{equation}
J_{AB}^{\mathcal{N}^{\Gamma}}=\left(J_{AB}^{\mathcal{N}}\right)^{T_{B}}.\label{idenppt}
\end{equation}
Furthermore, PPT maps have the property that they are completely PPT
preserving \citep{Chitambar-2018b}, meaning that if $\mathcal{N}_{AB}\in\mathrm{PPT}\left(A_{0}B_{0}\to A_{1}B_{1}\right)$,
then for every bipartite PPT quantum state $\rho\in\mathfrak{D}\left(A_{0}'A_{0}B_{0}'B_{0}\right)$,
the matrix $\mathcal{N}_{A_{0}B_{0}\to A_{1}B_{1}}\left(\rho_{A_{0}'A_{0}B_{0}'B_{0}}\right)$
has positive partial transpose. In other words, $\mathcal{N}_{AB}$
takes PPT positive semi-definite matrices to PPT positive semi-definite
matrices even when it is tensored with the identity. 

Here we discuss two types of generalizations of PPT maps to supermaps.
We call the first one \emph{restricted} PPT superchannels, to distinguish
it from the PPT supermaps we will study extensively in what follows.
We will see that restricted PPT superchannels lead to a cumbersome
entanglement theory on bipartite channels, similar to the one used
in Refs.~\citep{WW18,Kaur2017,Wilde-PPT-Das1}. Further, here we
consider bipartite channels, whereas in Ref.~\citep{WW18} the authors
considered only one-way channels from Alice to Bob (i.e.\ the special
case in which $\left|B_{0}\right|=\left|A_{1}\right|=1$).

A restricted PPT superchannel is depicted in Fig.~\ref{fig:PPT}.
\begin{figure}
\begin{centering}
\includegraphics[width=1\columnwidth]{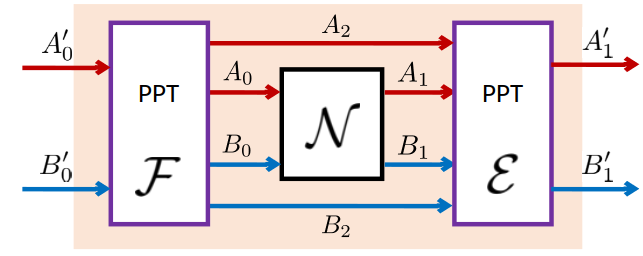}
\par\end{centering}
\caption{\label{fig:PPT}The action of a restricted PPT superchannel on the
bipartite channel $\mathcal{N}_{AB}$. }

\end{figure}
 In the language of Ref.~\citep{Gour-Scandolo-resource}, it is a
freely realizable superchannel: it consists of pre- and post- processing
channels $\mathcal{E}$ and $\mathcal{F}$ that are \emph{both} PPT.
Note that, at a first glance, this looks a very natural definition,
and as discussed in Ref.~\citep{Gour-Scandolo-resource}, it is the
most physical and natural one. Moreover, denoting this restricted
PPT superchannel by $\Theta$, it is clear that if $\mathcal{N}$
is a PPT channel then also the resulting map $\Theta\left[\mathcal{N}\right]$
is PPT. Nonetheless, PPT channels are not physical. They do not arise
from some \emph{physical} constraint on a physical system. Therefore,
the requirement that the superchannel $\Theta$ be realized with PPT
pre- and post-processing channels does not make $\Theta$ more physical.
Moreover, as we will see, this definition does not lead to a simple
resource theory, and as such, it loses its advantage of being a useful
approximation to LOCC. For these reasons, we will adopt a more general
definition of PPT superchannels that avoids the requirement that they
be realized by PPT channels. However, before doing that, we first
discuss some properties of restricted PPT superchannels. 
\begin{prop}
Let $\Theta\in\mathfrak{S}\left(AB\to A'B'\right)$ be a superchannel
as in Fig.~\ref{fig:PPT}, where $\mathcal{F}\in\mathrm{PPT}\left(A'_{0}B'_{0}\to A_{2}A_{0}B_{0}B_{2}\right)$
and $\mathcal{E}\in\mathrm{PPT}\left(A_{2}A_{1}B_{1}B_{2}\to A_{1}'B_{1}'\right)$.
Then 
\[
\left(\mathbf{J}_{ABA'B'}^{\Theta}\right)^{T_{BB'}}\geq0.
\]
\end{prop}

\begin{proof}
Since $\mathbf{J}_{ABA'B'}^{\Theta}$ is the Choi matrix of the CPTP
map
\begin{align}
 & \mathcal{Q}_{A_{1}A_{0}'B_{1}B_{0}'\to A_{0}A_{1}'B_{0}B_{1}'}^{\Theta}\nonumber \\
 & =\mathcal{E}_{A_{2}A_{1}B_{1}B_{2}\to A_{1}'B_{1}'}\circ\mathcal{F}_{A'_{0}B'_{0}\to A_{2}A_{0}B_{0}B_{2}},\label{q}
\end{align}
it is enough to show that the channel $\mathcal{Q}^{\Theta}$ is PPT.
Now, $\mathcal{Q}^{\Theta}$ is PPT because it is defined as a composition
of two PPT maps. Explicitly, we have
\begin{align*}
 & \left(\mathcal{Q}^{\Theta}\right)^{\Gamma}\\
 & =\mathcal{T}_{B_{0}B_{1}'}\circ\mathcal{E}_{A_{2}A_{1}B_{1}B_{2}\to A_{1}'B_{1}'}\circ\mathcal{F}_{A'_{0}B'_{0}\to A_{2}A_{0}B_{0}B_{2}}\circ\mathcal{T}_{B_{1}B_{0}'}\\
 & =\mathcal{T}_{B_{1}'}\circ\mathcal{E}_{A_{2}A_{1}B_{1}B_{2}\to A_{1}'B_{1}'}\circ\mathcal{T}_{B_{1}}\\
 & \circ\mathcal{T}_{B_{0}}\circ\mathcal{F}_{A'_{0}B'_{0}\to A_{2}A_{0}B_{0}B_{2}}\circ\mathcal{T}_{B_{0}'}\\
 & =\mathcal{T}_{B_{1}'}\circ\mathcal{E}_{A_{2}A_{1}B_{1}B_{2}\to A_{1}'B_{1}'}\circ\mathcal{T}_{B_{1}B_{2}}\\
 & \circ\mathcal{T}_{B_{0}B_{2}}\circ\mathcal{F}_{A'_{0}B'_{0}\to A_{2}A_{0}B_{0}B_{2}}\circ\mathcal{T}_{B_{0}'}\\
 & =\mathcal{E}_{A_{2}A_{1}B_{1}B_{2}\to A_{1}'B_{1}'}^{\Gamma}\circ\mathcal{F}_{A'_{0}B'_{0}\to A_{2}A_{0}B_{0}B_{2}}^{\Gamma}.
\end{align*}
Since both $\mathcal{E}$ and $\mathcal{F}$ are PPT channels, the
last line is a valid quantum channel. This completes the proof.
\end{proof}
We believe that the converse of the proposition above does \emph{not}
hold. In other words, if the Choi matrix of $\Theta$ has positive
partial transpose, it does not necessarily mean that $\Theta$ can
be realized with pre- and post-processing channels that are both PPT.
However, to prove such a statement, one will need to provide an example,
and then show that the proposed superchannel does not have any other
realizations that involve only PPT pre- and post-processing channels.
Alternatively, the question can be rephrased as follows. Suppose we
only know that the channel $\mathcal{Q}^{\Theta}$ in the first line
of Eq.~\eqref{q} is a PPT channel; does it necessarily mean that
there exist PPT channels $\mathcal{E}$ and $\mathcal{F}$ such that
we can decompose $\mathcal{Q}^{\Theta}$ as in the second line of
Eq.\ \eqref{q}?

While there are no obvious reasons to believe that the answer is positive,
we have not been able to prove it. If, instead, the answer were positive,
it would mean that the set of restricted PPT superchannels is the
same as the set of PPT superchannels we define below. 

\subsection{PPT supermaps\label{subsec:PPT-supermaps}}

In this section we define the set of PPT superchannels we are going
to use in the following \citep{Leung}. These superchannels have already
featured in a number of works on quantum communication \citep{XinWang1,XinWang2,XinWang3}.
We believe that this set is strictly larger than the set of restricted
PPT superchannels introduced above. However, as we discussed above,
we have not been able to show this strict inclusion.
\begin{defn}
\label{def:maindef}Let $\Theta\in\mathfrak{L}\left(AB\to A'B'\right)$
be a CP supermap with systems $A,B,A',B'$ all being composite systems
with input and output dimensions.
\begin{enumerate}
\item $\Theta$ is \emph{PPT-preserving} if for any PPT map $\mathcal{E}\in\mathrm{PPT}\left(A_{0}B_{0}\to A_{1}B_{1}\right)$,
the map $\Theta\left[\mathcal{E}\right]\in\mathrm{PPT}\left(A_{0}'B_{0}'\to A_{1}'B_{1}'\right)$.
\item $\Theta$ is \emph{completely PPT-preserving} if $\mathbbm{1}_{A''B''}\otimes\Theta$
is PPT preserving for any composite systems $A''=\left(A''_{0},A''_{1}\right)$
and $B''=\left(B_{0}'',B_{1}''\right)$.
\item $\Theta$ is a\emph{ PPT supermap} if, in addition to $\Theta$, also
$\Theta^{\Gamma}:=\Upsilon_{B'}\circ\Theta\circ\Upsilon_{B}$ is a
CP supermap.
\end{enumerate}
\end{defn}

\begin{rem}
Note that if $\Theta\in\mathfrak{L}\left(AB\to A'B'\right)$ is a
PPT CP supermap, and the dimensions $\left|A_{0}\right|=\left|B_{0}\right|=\left|A_{1}\right|=\left|B_{1}\right|=1$,
then $\Theta$ can be viewed as a PPT map in $\mathrm{CP}\left(A_{0}'B_{0}'\to A_{1}'B_{1}'\right)$.

Moreover, note that in the definition of a PPT supermap we require
that both $\Theta$ and $\Theta^{\Gamma}$ are CP supermaps, in complete
analogy with the definition of PPT CP maps.
\end{rem}

We denote the set of all PPT CP supermaps by $\mathrm{PPT}\left(AB\to A'B'\right)$.

The landscape of PPT supermaps portrayed in definition~\ref{def:maindef}
is actually simpler. Indeed, completely PPT-preserving and PPT supermaps
are the same notion (cf.\ also Ref.~\citep{Leung}).
\begin{thm}
Let $\Theta\in\mathfrak{L}\left(AB\to A'B'\right)$ be a CP supermap,
and denote by $\mathbf{J}_{ABA'B'}^{\Theta}$ its Choi matrix. Then,
the following are equivalent:
\begin{enumerate}
\item $\Theta$ is a PPT supermap.\label{enu:PPT1}
\item The Choi matrix of $\Theta$ satisfies\label{enu:PPT2}
\[
\left(\mathbf{J}_{ABA'B'}^{\Theta}\right)^{T_{BB'}}\geq0.
\]
\item $\Theta$ is completely PPT-preserving.\label{enu:PPT3}
\end{enumerate}
\end{thm}

\begin{proof}
First we prove that statements~\ref{enu:PPT1} and \ref{enu:PPT2}
are equivalent. Consider the map $\Phi_{A\widetilde{A}}^{+}\in\mathrm{CP}\left(A_{0}\widetilde{A}_{0}\to A_{1}\widetilde{A}_{1}\right)$
defined in Eq.~\eqref{eq:maxent}, which is completely positive,
and it is the CP-map analog of the maximally entangled state. Recall
also that one of the representations of a supermap $\Theta\in\mathfrak{L}\left(A\to A'\right)$,
is given by the map $\mathcal{P}_{AA'}^{\Theta}=\left(\mathbbm{1}_{A}\otimes\Theta\right)\left[\Phi_{A\widetilde{A}}^{+}\right]$
whose Choi matrix is the Choi matrix of $\Theta$. Since here we consider
a bipartite CP supermap $\Theta\in\mathfrak{L}\left(AB\to A'B'\right)$,
the map $\mathcal{P}^{\Theta}$ is defined as 
\[
\mathcal{P}_{ABA'B'}^{\Theta}=\left(\mathbbm{1}_{AB}\otimes\Theta_{\widetilde{A}\widetilde{B}\to A'B'}\right)\left[\Phi_{A\widetilde{A}}^{+}\otimes\Phi_{B\widetilde{B}}^{+}\right],
\]
where we have used the fact that $\Phi_{AB\widetilde{A}\widetilde{B}}^{+}=\Phi_{A\widetilde{A}}^{+}\otimes\Phi_{B\widetilde{B}}^{+}$.
Now, observe that
\begin{align}
 & \mathcal{P}_{ABA'B'}^{\Theta^{\Gamma}}\nonumber \\
 & =\left(\mathbbm{1}_{AB}\otimes\Upsilon_{B'}\circ\Theta_{\widetilde{A}\widetilde{B}\to A'B'}\circ\Upsilon_{\widetilde{B}}\right)\left[\Phi_{A\widetilde{A}}^{+}\otimes\Phi_{B\widetilde{B}}^{+}\right]\label{eq:lambda}
\end{align}
and 
\begin{equation}
\Upsilon_{\widetilde{B}}\left[\Phi_{B\widetilde{B}}^{+}\right]=\mathcal{T}_{\widetilde{B}_{1}}\circ\Phi_{B\widetilde{B}}^{+}\circ\mathcal{T}_{\widetilde{B}_{0}}=\Upsilon_{B}\left[\Phi_{B\widetilde{B}}^{+}\right],\label{eq:PPT Phi+}
\end{equation}
where in the last equality we used the representation~\eqref{eq:maxent}
of $\Phi_{B\widetilde{B}}^{+}$, and the fact that $\left(\phi_{B_{0}\widetilde{B}_{0}}^{+}\right)^{T_{\widetilde{B}_{0}}}=\left(\phi_{B_{0}\widetilde{B}_{0}}^{+}\right)^{T_{B_{0}}}$
and $\left(\phi_{B_{1}\widetilde{B}_{1}}^{+}\right)^{T_{\widetilde{B}_{1}}}=\left(\phi_{B_{1}\widetilde{B}_{1}}^{+}\right)^{T_{B_{1}}}$.
Combining this with Eq.~\eqref{eq:lambda}, we conclude that
\begin{align}
 & \mathcal{P}_{ABA'B'}^{\Theta^{\Gamma}}\nonumber \\
 & =\left(\Upsilon_{B'}\otimes\Upsilon_{B}\right)\circ\left(\mathbbm{1}_{AB}\otimes\Theta_{\widetilde{A}\widetilde{B}\to A'B'}\right)\left[\Phi_{A\widetilde{A}}^{+}\otimes\Phi_{B\widetilde{B}}^{+}\right]\nonumber \\
 & =\left(\mathcal{P}_{ABA'B'}^{\Theta}\right)^{\Gamma}.\label{eq:transpose}
\end{align}
Hence 
\[
\mathbf{J}_{ABA'B'}^{\Theta^{\Gamma}}=\left(\mathbf{J}_{ABA'B'}^{\Theta}\right)^{T_{BB'}},
\]
where we have used Eq.~\eqref{idenppt}. This completes the proof
of the equivalence between statements \ref{enu:PPT1} and \ref{enu:PPT2}.

For the equivalence between~\ref{enu:PPT1} and \ref{enu:PPT3},
let $\Theta\in\mathfrak{L}\left(AB\to A'B'\right)$ be a PPT supermap.
Then, for any systems $A''B''$ and any PPT bipartite CP map, $\mathcal{N}_{A''B''AB}$,
we have
\begin{align*}
0 & \leq\Theta_{AB\to A'B'}^{\Gamma}\left[\Upsilon_{B''B}\left[\mathcal{N}_{A''B''AB}\right]\right]\\
 & =\Upsilon_{B''B'}\left[\Theta_{AB\to A'B'}\left[\mathcal{N}_{A''B''AB}\right]\right],
\end{align*}
where the equality follows from the definition of $\Theta_{AB\to A'B'}^{\Gamma}$.
In other words, $\left(\mathbbm{1}_{A''B''}\otimes\Theta\right)\left[\mathcal{N}_{A''B''AB}\right]$
is a PPT map, so $\Theta$ is completely PPT preserving.

Conversely, let $\Theta\in\mathfrak{L}(\widetilde{A}\widetilde{B}\to A'B')$
be a CP supermap that is completely PPT preserving. Note that, by
Eq.~\eqref{eq:PPT Phi+}, $\Phi_{AB\widetilde{A}\widetilde{B}}^{+}=\Phi_{A\widetilde{A}}^{+}\otimes\Phi_{B\widetilde{B}}^{+}$
is a PPT map. Therefore, the CP map 
\[
\mathcal{P}_{ABA'B'}^{\Theta}=\left(\mathbbm{1}_{AB}\otimes\Theta_{\widetilde{A}\widetilde{B}\to A'B'}\right)\left[\Phi_{A\widetilde{A}}^{+}\otimes\Phi_{B\widetilde{B}}^{+}\right]
\]
is PPT. From a similar relation to Eq.~\eqref{eq:transpose}, it
follows that $\mathcal{P}_{ABA'B'}^{\Theta^{\Gamma}}\geq0$, so $\Theta$
is PPT. This completes the proof.
\end{proof}
We end this section with a convenient property of the partial transpose
operation. This will be very useful in the following.
\begin{prop}
\label{proppt}Let $\Theta\in\mathfrak{L}\left(AB\to A'B'\right)$
be a bipartite supermap and let $\mathcal{N}\in\mathfrak{L}\left(A_{0}B_{0}\to A_{1}B_{1}\right)$
be a bipartite map. Then, 
\[
\left(\Theta\left[\mathcal{N}\right]\right)^{\Gamma}=\Theta^{\Gamma}\left[\mathcal{N}^{\Gamma}\right].
\]
\end{prop}

\begin{proof}
Note that
\begin{align*}
\Theta^{\Gamma}\left[\mathcal{N}^{\Gamma}\right] & =\Upsilon_{B'}\circ\Theta\circ\Upsilon_{B}\left[\Upsilon_{B}\left[\mathcal{N}_{AB}\right]\right]\\
 & =\Upsilon_{B'}\circ\Theta\left[\mathcal{N}_{AB}\right]\\
 & =\left(\Theta\left[\mathcal{N}_{AB}\right]\right)^{\Gamma}.
\end{align*}
This completes the proof.
\end{proof}

\subsection{Single-shot interconversions\label{subsec:Single-shot-interconversions}}

In the resource theory of NPT static entanglement the conversion of
one resource into another can be characterized by SDPs \citep{Fang2019}.
Here we show that for NPT dynamical entanglement, the conversion distance
$d_{\mathrm{PPT}}\left(\mathcal{N}_{AB}\to\mathcal{M}_{A'B'}\right)$,
defined in Ref.~\citep{Gour-Scandolo-resource}, can be computed
by an SDP as long as we consider the PPT superchannels introduced
in definition~\ref{def:maindef}, and \emph{not} the restricted PPT
operations illustrated in Fig.~\ref{fig:PPT}.

Now, recall that in NPT entanglement theory, the conversion distance
is
\begin{align*}
 & d_{\mathrm{PPT}}\left(\mathcal{N}_{AB}\to\mathcal{M}_{A'B'}\right)\\
 & =\frac{1}{2}\min_{\Theta\in\mathrm{PPT}\left(AB\to A'B'\right)}\left\Vert \Theta_{AB\to A'B'}\left[\mathcal{N}_{AB}\right]-\mathcal{M}_{A'B'}\right\Vert _{\diamond}.
\end{align*}
Now, the diamond norm can be expressed as the SDP \citep{Watrous-diamond}
\begin{align*}
 & \frac{1}{2}\left\Vert \Theta_{AB\to A'B'}\left[\mathcal{N}_{AB}\right]-\mathcal{M}_{A'B'}\right\Vert _{\diamond}\\
 & =\min_{\omega_{A'B'}\geq0;\thinspace\omega_{A'B'}\geq J_{A'B'}^{\Theta\left[\mathcal{N}\right]-\mathcal{M}}}\left\Vert \omega_{A'_{0}B'_{0}}\right\Vert _{\infty}.
\end{align*}
Now, in Ref.~\citep{Gour-Winter}, it was shown that it can be written
also as
\begin{align*}
 & \frac{1}{2}\left\Vert \Theta_{AB\to A'B'}\left[\mathcal{N}_{AB}\right]-\mathcal{M}_{A'B'}\right\Vert _{\diamond}\\
 & =\min\left\{ \lambda:\lambda\mathcal{Q}_{A'B'}\geq\Theta_{AB\to A'B'}\left[\mathcal{N}_{AB}\right]-\mathcal{M}_{A'B'}\right\} ,
\end{align*}
where $\mathcal{Q}_{B}\in\mathrm{CPTP}\left(B_{0}\rightarrow B_{1}\right)$.
Therefore, calculating the conversion distance amounts to solving
the following minimization problem 
\begin{eqnarray}
\textrm{Find} & \quad & d_{\mathrm{PPT}}\left(\mathcal{N}_{AB}\to\mathcal{M}_{A'B'}\right)=\min\lambda\nonumber \\
\textrm{Subject to:} & \quad & \lambda\mathcal{Q}_{A'B'}\geq\Theta_{AB\to A'B'}\left[\mathcal{N}_{AB}\right]-\mathcal{M}_{A'B'}\nonumber \\
 & \quad & \mathcal{Q}\textrm{ channel}\nonumber \\
 & \quad & \Theta\textrm{ superchannel}.\label{eq:primalppt}
\end{eqnarray}
This can be rephrased as the following SDP form. Denote the Choi matrix
of $\lambda\mathcal{Q}_{A'B'}$ by $\alpha_{A'B'}$, and the Choi
matrix of $\Theta$ by $\mathbf{J}_{ABA'B'}$. Then, following Ref.~\citep{Gour-Scandolo-resource},
we can express Eq.~\eqref{eq:primalppt} in terms of Choi matrices,
finding the new optimization problem
\[
d_{\mathrm{PPT}}\left(\mathcal{N}_{AB}\to\mathcal{M}_{A'B'}\right)=\frac{1}{\left|A_{0}'B_{0}'\right|}\min\mathrm{Tr}\left[\alpha_{A'B'}\right]
\]
subject to
\begin{align*}
 & \alpha_{A'B'}\geq0\\
 & \alpha_{A_{0}'B_{0}'}=\mathrm{Tr}\left[\alpha_{A_{0}'B_{0}'}\right]u_{A_{0}'B_{0}'}\\
 & \alpha\geq\mathrm{Tr}_{AB}\left[\mathbf{J}_{ABA'B'}\left(\left(J_{AB}^{\mathcal{N}}\right)^{T}\otimes I_{A'B'}\right)\right]-J_{A'B'}^{\mathcal{M}}\\
 & \mathbf{J}_{ABA'B'}\geq0\\
 & \mathbf{J}_{ABA_{0}'B_{0}'}=\mathbf{J}_{A_{0}B_{0}A_{0}'B_{0}'}\otimes u_{A_{1}B_{1}}\\
 & \mathbf{J}_{A_{1}B_{1}A_{0}'B_{0}'}=I_{A_{1}B_{1}A_{0}'B_{0}'}\\
 & \mathbf{J}_{ABA'B'}^{T_{BB'}}\geq0.
\end{align*}
Clearly, the above optimization can be solved efficiently and algorithmically
with an SDP. We can also express it in its dual form following Ref.~\citep{Gour-Scandolo-resource}:
\begin{align}
 & d_{\mathrm{PPT}}\left(\mathcal{N}_{AB}\to\mathcal{M}_{A'B'}\right)\nonumber \\
 & =\max\left\{ t\left|A_{1}B_{1}A_{0}'B_{0}'\right|-\mathrm{Tr}\left[\zeta_{A'B'}J_{A'B'}^{\mathcal{M}}\right]\right\} ,\label{eq:dual PPT 1}
\end{align}
subject to
\begin{align}
 & \left(J_{AB}^{\mathcal{N}}\right)^{T}\otimes\zeta_{A'B'}-tI_{ABA'B'}\in\mathfrak{J}_{ABA'B'}^{*}\nonumber \\
 & 0\leq\zeta_{A'B'}\leq\eta_{A_{0}'B_{0}'}\otimes I_{A_{1}'B_{1}'}\nonumber \\
 & \mathrm{Tr}\left[\eta_{A_{0}'B_{0}'}\right]=1,\label{eq:dual PPT 1'}
\end{align}
where $\mathfrak{J}_{ABA'B'}^{*}$ is defined in Eq.~\eqref{eq:dual cone PPT}.
We want to show that this dual problem is an SDP as well. To this
end, from Eq.~\eqref{eq:form}, define
\begin{align*}
 & \beta_{ABA'_{0}B'_{0}}\\
 & :=\frac{1}{\left|A_{0}B_{0}\right|}\left(tI_{ABA'_{0}B'_{0}}+Y_{ABA_{0}'B_{0}'}+I_{A_{0}B_{0}}\otimes Z_{A_{1}B_{1}A_{0}'B_{0}'}\right),
\end{align*}
where, like in Eq.~\eqref{eq:form}, $Y$ is a Hermitian matrix such
that $Y_{AB}=0$, and $Z$ is a Hermitian matrix such that $\mathrm{Tr}\left[Z_{A_{1}B_{1}A_{0}'B_{0}'}\right]=0$.
In this way, by the definition of $\beta_{ABA'_{0}B'_{0}}$, and recalling
Eq.~\eqref{eq:form}, we can rewrite Eqs.~\eqref{eq:dual PPT 1}
and \eqref{eq:dual PPT 1'} as
\begin{align*}
 & d_{\mathrm{PPT}}\left(\mathcal{N}_{AB}\to\mathcal{M}_{A'B'}\right)\\
 & =\max\left\{ \mathrm{Tr}\left[\beta_{ABA_{0}'B_{0}'}\right]-\mathrm{Tr}\left[\zeta_{A'B'}J_{A'B'}^{\mathcal{M}}\right]\right\} ,
\end{align*}
subject to
\begin{align*}
 & \beta\in\mathrm{Herm}\left(ABA_{0}'B_{0}'\right)\\
 & \beta_{AB}=u_{A_{0}B_{0}}\otimes\beta_{A_{1}B_{1}}\\
 & 0\leq\zeta_{A'B'}\leq\eta_{A_{0}'B_{0}'}\otimes I_{A_{1}'B_{1}'}\\
 & \mathrm{Tr}\left[\eta_{A_{0}'B_{0}'}\right]=1\\
 & \left(\left(J_{AB}^{\mathcal{N}}\right)^{T}\otimes\zeta_{A'B'}-\left|A_{0}B_{0}\right|\beta\otimes I_{B_{1}'A_{1}'}-P_{ABA'B'}\right)^{T_{BB'}}\geq0\\
 & P_{ABA'B'}\geq0.
\end{align*}
Hence, the computation of $d_{\mathrm{PPT}}\left(\mathcal{N}_{AB}\to\mathcal{M}_{A'B'}\right)$
in the resource theory of NPT entanglement is an SDP optimization
problem. We point out that if we considered restricted PPT superchannels,
instead, the condition that $\Theta$ is free would be expressed as
the condition that the CPTP map $\mathcal{Q}^{\Theta}$ has a decomposition
into pre- and post-processing that are both PPT channels, like in
Eq.~\eqref{q}. This condition appears to be very cumbersome, and
it is not clear if the determination of whether or not $\Theta$ has
the form~\eqref{q} can be solved with an SDP.

\subsection{NPT entanglement measures\label{subsec:NPT-entanglement-measures}}

In the entanglement theory for static resources, functions that behave
monotonically under PPT operations, also behave monotonically under
LOCC operations, as LOCC is a subset of PPT. Hence, any NPT entanglement
measure is also an LOCC entanglement measure. The advantage of some
of the NPT entanglement measures is that they can be computed with
SDPs (see e.g.\ the family of measures discussed in Ref.~\citep{WW19}).
In this section we study a few of these measures.

\subsubsection{Negativity and logarithmic negativity of bipartite channels}

A well-known NPT entanglement measure is the negativity \citep{VW2002}.
It is defined on a bipartite quantum state $\rho\in\mathfrak{D}\left(A_{0}B_{0}\right)$
as 
\[
N\left(\rho_{A_{0}B_{0}}\right)=\frac{\left\Vert \mathcal{T}_{B_{0}}\left(\rho_{A_{0}B_{0}}\right)\right\Vert _{1}-1}{2}.
\]
The generalization of the negativity to bipartite channels can be
done by replacing the input bipartite state $\rho\in\mathfrak{D}\left(A_{0}B_{0}\right)$
with input bipartite channel $\mathcal{N}\in\mathrm{CPTP}\left(A_{0}B_{0}\to A_{1}B_{1}\right)$,
the trace norm with the diamond norm, and the transpose map $\mathcal{T}_{B_{0}}$
with the transpose supermap $\Upsilon_{B}$. The negativity of the
bipartite channel is therefore defined as
\[
N\left(\mathcal{N}_{AB}\right)=\frac{\left\Vert \Upsilon_{B}\left[\mathcal{N}_{AB}\right]\right\Vert _{\diamond}-1}{2}.
\]
 Furthermore, the logarithmic negativity is defined as
\[
LN\left(\mathcal{N}_{AB}\right)=\log_{2}\left\Vert \Upsilon_{B}\left[\mathcal{N}_{AB}\right]\right\Vert _{\diamond}.
\]

To show that the above quantities are indeed good generalizations
of the negativity and logarithmic negativity to bipartite channels,
we show that they vanish on PPT bipartite channels, and that they
behave monotonically under PPT superchannels. They vanish on PPT bipartite
channels because if $\mathcal{N}_{AB}$ is PPT then $\Upsilon_{B}\left[\mathcal{N}_{AB}\right]$
is a quantum channel so its diamond norm is 1. To show the monotonicity
property, let $\Theta\in\mathrm{PPT}\left(AB\to A'B'\right)$ and
observe that
\begin{align*}
\left\Vert \Upsilon_{B'}\circ\Theta\left[\mathcal{N}_{AB}\right]\right\Vert _{\diamond} & =\left\Vert \left(\Theta\left[\mathcal{N}_{AB}\right]\right)^{\Gamma}\right\Vert _{\diamond}\\
 & =\left\Vert \Theta^{\Gamma}\left[\mathcal{N}_{AB}^{\Gamma}\right]\right\Vert _{\diamond}\\
 & \leq\left\Vert \Upsilon_{B}\left[\mathcal{N}_{AB}\right]\right\Vert _{\diamond},
\end{align*}
where in the first equality we used the definition of the partial
transpose of a channel, in the second proposition~\ref{proppt},
and finally the inequality follows from the fact that $\Theta^{\Gamma}$
is a superchannel because $\Theta$ is a PPT superchannel, and the
fact that the diamond norm is contractive under superchannels \citep{Gour2018}.
Therefore, since both the negativity and the logarithmic negativity
are increasing functions of $\left\Vert \Upsilon_{B}\left[\mathcal{N}_{AB}\right]\right\Vert _{\diamond}$,
we conclude that they are non-increasing under PPT superchannels.

\subsubsection{A complete set of computationally manageable measures of bipartite
NPT dynamical entanglement}

We can use the same technique as above to generalize other measures
of NPT static entanglement to NPT dynamical entanglement (see e.g.\ Ref.~\citep{VW2002}).
Now we focus on the complete family of measures introduced in Ref.~\citep{Gour-Scandolo-resource}.
In the case of NPT entanglement, for any bipartite channel $\mathcal{P}_{A'B'}\in\mathrm{CPTP}\left(A_{0}'B_{0}'\to A_{1}'B_{1}'\right)$,
using the results in Ref.~\citep{Gour-Scandolo-resource}, we can
write 
\[
f_{\mathcal{P}}\left(\mathcal{N}_{AB}\right)=\max_{\mathbf{J}\in\mathfrak{J}}\mathrm{Tr}\left[\mathbf{J}_{ABA'B'}\left(\left(J_{AB}^{\mathcal{N}}\right)^{T}\otimes J_{A'B'}^{\mathcal{P}}\right)\right]
\]
for every quantum channel $\mathcal{N}_{AB}$, where $\mathfrak{J}$
is the set of Choi matrices of PPT superchannels (note that it is
compact and convex). In other words, $\mathbf{J}_{ABA'B'}$ is subject
to the following constraints: 
\begin{enumerate}
\item $\mathbf{J}_{ABA'B'}\geq0$; $\mathbf{J}_{ABA'_{0}B_{0}'}=\mathbf{J}_{A_{0}B_{0}A_{0}'B_{0}'}\otimes u_{A_{1}B_{1}}$;
$\mathbf{J}_{A_{1}B_{1}A_{0}'B_{0}'}=I_{A_{1}B_{1}A_{0}'B_{0}'}$;
\item $\mathbf{J}_{ABA'B'}^{T_{BB'}}\geq0$.
\end{enumerate}
The first group of conditions above ensures that $\mathbf{J}_{ABA'B'}$
is the Choi matrix of a superchannel in $\mathfrak{S}\left(AB\to A'B'\right)$;
the second condition guarantees that the superchannel is free, i.e.\ PPT.
A key observation about the above optimization problem is that it
is an SDP. As noted in Ref.~\citep{Gour-Scandolo-resource}, the
family of convex functions $\left\{ f_{\mathcal{P}}\right\} $, indexed
by all $\mathcal{P}\in\mathrm{CPTP}\left(A_{0}'B_{0}'\to A_{1}'B_{1}'\right)$,
is complete, in the sense that there exists a PPT superchannel converting
a bipartite channel $\mathcal{N}\in\mathrm{CPTP}\left(A_{0}B_{0}\to A_{1}B_{1}\right)$
into another bipartite channel $\mathcal{E}\in\mathrm{CPTP}\left(A_{0}'B_{0}'\to A_{1}'B_{1}'\right)$
if and only if 
\begin{equation}
f_{\mathcal{P}}\left(\mathcal{N}_{AB}\right)\geq f_{\mathcal{P}}\left(\mathcal{E}_{A'B'}\right)\label{formmonotone}
\end{equation}
for every $\mathcal{P}\in\mathrm{CPTP}\left(A_{0}'B_{0}'\to A_{1}'B_{1}'\right)$.

One may argue that the above condition cannot be checked efficiently,
as it involves an (uncountably) infinite number of measures of dynamical
entanglement, labeled by \emph{all} quantum channels $\mathcal{P}$.
However, we have another way to determine whether two bipartite quantum
channels $\mathcal{N}_{AB}$ and $\mathcal{E}_{A'B'}$ can be interconverted
by PPT superchannels, which is to compute the conversion distance
$d_{\mathrm{PPT}}\left(\mathcal{N}_{AB}\to\mathcal{E}_{A'B'}\right)$:
$d_{\mathrm{PPT}}\left(\mathcal{N}_{AB}\to\mathcal{E}_{A'B'}\right)=0$
if and only if $\mathcal{N}_{AB}$ can be converted (exactly) into
$\mathcal{E}_{A'B'}$ by PPT superchannels. In section~\ref{subsec:Single-shot-interconversions}
we showed that this can be done efficiently with an SDP.

Why do we consider this family of dynamical entanglement measures,
then? Their significance is that they completely characterize the
NPT entanglement of a \emph{single} bipartite channel, whereas the
computation of $d_{\mathrm{PPT}}\left(\mathcal{N}_{AB}\to\mathcal{E}_{A'B'}\right)$
requires to know \emph{both} of its inputs $\mathcal{N}$ and $\mathcal{E}$,
i.e.\ also the target channel. Hence, Eq.~\eqref{formmonotone}
demonstrates that the convertibility can be expressed in a monotonic
form, similarly to Vidal's monotones \citep{Nielsen,Vidal,Jonathan}
in the theory of pure-state bipartite entanglement.
\begin{rem}
If we want measures of NPT dynamical entanglement that vanish on PPT
channels, we can consider the measures $G_{\mathcal{P}}\left(\mathcal{N}_{AB}\right)=f_{\mathcal{P}}\left(\mathcal{N}_{AB}\right)-\max_{\mathcal{M}}\mathrm{Tr}\left[J_{A'B'}^{\mathcal{M}}J_{A'B'}^{\mathcal{P}}\right]$.
Here $\mathcal{M}$ ranges over all PPT channels (again, a compact
and convex set).
\end{rem}

\subsubsection{The max-logarithmic negativity}

In Ref.~\citep{WW18} the authors considered a measure of NPT entanglement,
which they called the $\kappa$-entanglement. For bipartite states,
it is defined as
\begin{align*}
 & E_{\kappa}\left(\rho_{AB}\right)\\
 & =\log_{2}\inf\left\lbrace \mathrm{Tr}\left[S_{AB}\right]:-S_{AB}^{T_{B}}\leq\rho_{AB}^{T_{B}}\leq S_{AB}^{T_{B}}\;;\;S_{AB}\geq0\right\rbrace ,
\end{align*}
and for one-way channels $\mathcal{E}_{A_{0}\rightarrow B_{1}}$ as
\begin{align*}
 & E_{\kappa}\left(\mathcal{E}_{A_{0}\rightarrow B_{1}}\right)\\
 & =\log_{2}\inf\left\lbrace \left\Vert J_{A_{0}}^{\mathcal{Q}}\right\Vert _{\infty}:-\mathcal{Q}^{\Gamma}\leq\mathcal{E}_{A_{0}\rightarrow B_{1}}^{\Gamma}\leq\mathcal{Q}^{\Gamma};\,\mathcal{Q}\geq0\right\rbrace .
\end{align*}
The significance of this measure is that it has an operational interpretation
as the exact asymptotic cost under PPT operations. Here we introduce
the max-logarithmic negativity (MLN) (see also Ref.~\citep{WW19}),
which has a similar operational interpretation, and is a generalization
of the $\kappa$-entanglement to bipartite channels. However, as we
will see, for bipartite channels, there are two possible generalizations
of the quantity given in Ref.~\citep{WW18}, and we define the MLN
to be the maximum of the two. Explicitly, the MLN is defined as 
\[
LN_{\max}\left(\mathcal{N}_{AB}\right)=\max\left\{ LN_{\max}^{\left(0\right)}\left(\mathcal{N}_{AB}\right),LN_{\max}^{\left(1\right)}\left(\mathcal{N}_{AB}\right)\right\} ,
\]
where
\begin{align*}
 & LN_{\max}^{\left(0\right)}\left(\mathcal{N}_{AB}\right)\\
 & =\log_{2}\inf\left\{ \left\Vert J_{A_{0}B_{0}}^{\mathcal{P}}\right\Vert _{\infty}:-\mathcal{P}_{AB}^{\Gamma}\leq\mathcal{N}_{AB}^{\Gamma}\leq\mathcal{P}_{AB}^{\Gamma};\,\mathcal{P}\geq0\right\} 
\end{align*}
and
\begin{align*}
 & LN_{\max}^{\left(1\right)}\left(\mathcal{N}_{AB}\right)\\
 & =\log_{2}\inf\left\{ \left\Vert J_{A_{0}B_{0}}^{\mathcal{P}^{\Gamma}}\right\Vert _{\infty}:-\mathcal{P}_{AB}^{\Gamma}\leq\mathcal{N}_{AB}^{\Gamma}\leq\mathcal{P}_{AB}^{\Gamma};\thinspace\mathcal{P}\geq0\right\} .
\end{align*}
The above quantities can be computed with SDP. In particular, they
have a dual, giving an alternative expression for them:
\begin{equation}
LN_{\max}^{\left(0\right)}\left(\mathcal{N}_{AB}\right)=\log_{2}\sup\left\{ \mathrm{Tr}\left[J_{AB}^{\mathcal{N}}\left(V_{AB}-W_{AB}\right)\right]\right\} \label{eq:MLN0}
\end{equation}
subject to
\begin{align*}
 & V_{AB}+W_{AB}\leq\rho_{A_{0}B_{0}}\otimes I_{A_{1}B_{1}}\\
 & \rho\in\mathfrak{D}\left(A_{0}B_{0}\right)\\
 & V_{AB}\geq0\\
 & W_{AB}\geq0
\end{align*}
and
\[
LN_{\max}^{\left(1\right)}\left(\mathcal{N}_{AB}\right)=\log_{2}\sup\left\{ \mathrm{Tr}\left[J_{AB}^{\mathcal{N}}\left(V_{AB}-W_{AB}\right)\right]\right\} 
\]
subject to
\begin{align*}
 & V_{AB}+W_{AB}\leq\rho_{A_{0}B_{0}}^{T_{B_{0}}}\otimes I_{A_{1}B_{1}}\\
 & \rho\in\mathfrak{D}\left(A_{0}B_{0}\right)\\
 & V_{AB}\geq0\\
 & W_{AB}\geq0.
\end{align*}

\noindent These expressions can be obtained with the usual SDP techniques.
By Sion's minimax theorem, we can swap the order of the infimum and
the maximum in the definition of the MLN, so
\[
LN_{\max}\left(\mathcal{N}_{AB}\right)=\log_{2}\inf\left\{ \max\left\{ \left\Vert J_{A_{0}B_{0}}^{\mathcal{P}}\right\Vert _{\infty},\left\Vert J_{A_{0}B_{0}}^{\mathcal{P}^{\Gamma}}\right\Vert _{\infty}\right\} \right\} 
\]
subject to $-\mathcal{P}_{AB}^{\Gamma}\leq\mathcal{N}_{AB}^{\Gamma}\leq\mathcal{P}_{AB}^{\Gamma}$
and $\mathcal{P}_{AB}\geq0$. The MLN is defined here in terms of
the bipartite map $\mathcal{P}\in\mathrm{CP}\left(A_{0}B_{0}\to A_{1}B_{1}\right)$.
Denoting its Choi matrix by $P_{AB}\in\mathrm{Herm}\left(AB\right)$,
we can express the MLN as
\begin{widetext}
\noindent 
\begin{equation}
LN_{\max}\left(\mathcal{N}_{AB}\right)=\log_{2}\inf\left\{ \max\left\{ \left\Vert P_{A_{0}B_{0}}\right\Vert _{\infty},\left\Vert P_{A_{0}B_{0}}^{T_{B_{0}}}\right\Vert _{\infty}\right\} :-P_{AB}^{T_{B}}\leq\left(J_{AB}^{\mathcal{N}}\right)^{T_{B}}\leq P_{AB}^{T_{B}},\,P_{AB}\geq0\right\} .\label{eq:MLN Choi}
\end{equation}
\end{widetext}

\noindent Now we show here that many properties of the $\kappa$-entanglement
discussed in Ref.~\citep{WW18} carry over to the max-logarithmic
negativity, including the operational meaning of single-shot exact
entanglement cost (cf.\ section~\ref{subsec:NPT cost}). Moreover,
we will see that the max-logarithmic negativity is monotonic under
PPT superchannels, which we believe is a strictly larger set than
the set discussed in Ref.~\citep{WW18}, that is the set of restricted
PPT superchannels, which can be implemented by PPT pre- and post-processing,
like in Fig.~\ref{fig:PPT}.

\subsubsection{Properties of the max-logarithmic negativity}

Here we list a few key properties of the MLN. The first two show that
it reduces to $E_{\kappa}$ introduced in Ref.~\citep{WW18}.

\paragraph{Reduction to $\kappa$-entanglement for states}

A bipartite state can be viewed as a bipartite channel $\mathcal{N}_{AB}$
with $\left|A_{0}\right|=\left|B_{0}\right|=1$. In this case, in
Eq.\ \eqref{eq:MLN Choi}, $P_{A_{0}B_{0}}=P_{A_{0}B_{0}}^{T_{B}}=\mathrm{Tr}\left[P_{A_{1}B_{1}}\right]$.
Recalling that $P_{A_{1}B_{1}}\geq0$, we have $LN_{\max}\left(\mathcal{N}_{AB}\right)=\log_{2}\inf\left\lbrace \mathrm{Tr}\left[P_{A_{1}B_{1}}\right]\right\rbrace $,
subject to $-P_{A_{1}B_{1}}^{T_{B_{1}}}\leq\rho_{A_{1}B_{1}}^{T_{B_{1}}}\leq P_{A_{1}B_{1}}^{T_{B_{1}}}$
and $P_{A_{1}B_{1}}\geq0$. This expression coincides with $E_{\kappa}\left(\rho_{A_{1}B_{1}}\right)$.

\paragraph{Reduction to $\kappa$-entanglement for\ one-way channels}

For $\left|B_{0}\right|=\left|A_{1}\right|=1$, the channel $\mathcal{N}_{AB}$
can be viewed as a map $\mathcal{E}\in\mathrm{CPTP}\left(A_{0}\to B_{1}\right)$
and 
\[
LN_{\max}\left(\mathcal{N}_{AB}\right)=E_{\kappa}\left(\mathcal{E}_{A_{0}\to B_{1}}\right).
\]

\paragraph{Monotonicity}

Let $\mathcal{N}\in\mathrm{CPTP}\left(A_{0}B_{0}\to A_{1}B_{1}\right)$
be a bipartite channel, and let $\Theta\in\mathrm{PPT}\left(AB\to A'B'\right)$
be a PPT superchannel. Then, 
\[
LN_{\max}\left(\Theta_{AB\to A'B'}\left[\mathcal{N}_{AB}\right]\right)\leq LN_{\max}\left(\mathcal{N}_{AB}\right).
\]

\begin{proof}
Recall that for any superchannel $\Theta$ and bipartite channel $\mathcal{N}_{AB}$
we have $\left(\Theta\left[\mathcal{N}_{AB}\right]\right)^{\Gamma}=\Theta^{\Gamma}\left[\mathcal{N}_{AB}^{\Gamma}\right]$
(see proposition~\ref{proppt}). Hence, from the expression
\begin{align*}
 & LN_{\max}\left(\Theta\left[\mathcal{N}_{AB}\right]\right)=\\
 & =\log_{2}\inf\left\{ \max\left\{ \left\Vert J_{A_{0}'B_{0}'}^{\mathcal{R}}\right\Vert _{\infty},\left\Vert \left(J_{A_{0}'B_{0}'}^{\mathcal{R}}\right)^{T_{B_{0}}}\right\Vert _{\infty}\right\} \right\} ,
\end{align*}
subject to $-\mathcal{R}_{A'B'}^{\Gamma}\leq\Theta^{\Gamma}\left[\mathcal{N}_{AB}^{\Gamma}\right]\leq\mathcal{R}_{A'B'}^{\Gamma}$
and $\mathcal{R}_{A'B'}\geq0$, we can definitely write
\begin{align}
 & LN_{\max}\left(\Theta\left[\mathcal{N}_{AB}\right]\right)\nonumber \\
 & \leq\log_{2}\inf\left\{ \max\left\{ \left\Vert J_{A_{0}'B_{0}'}^{\Theta\left[\mathcal{P}\right]}\right\Vert _{\infty},\left\Vert J_{A_{0}'B_{0}'}^{\Theta^{\Gamma}\left[\mathcal{P}^{\Gamma}\right]}\right\Vert _{\infty}\right\} \right\} ,\label{eq:ineq1}
\end{align}
where $-\left(\Theta\left[\mathcal{P}_{AB}\right]\right)^{\Gamma}\leq\Theta^{\Gamma}\left[\mathcal{N}_{AB}^{\Gamma}\right]\leq\left(\Theta\left[\mathcal{P}_{AB}\right]\right)^{\Gamma}$
and $\mathcal{P}_{AB}\geq0$. Indeed, this inequality follows because
we have restricted $\mathcal{R}_{A'B'}$ to CP maps of the form $\Theta_{AB\to A'B'}\left[\mathcal{P}_{AB}\right]$,
where $\mathcal{P}\in\mathrm{CP}\left(A_{0}B_{0}\to A_{1}B_{1}\right)$,
and $\Theta$ is a PPT superchannel. Next, observe that, by the properties
of the Choi matrices of superchannels, 
\begin{align*}
 & J_{A_{0}'B_{0}'}^{\Theta\left[\mathcal{P}\right]}\\
 & =\mathrm{Tr}_{ABA_{1}'B_{1}'}\left[\mathbf{J}_{ABA'B'}^{\Theta}\left(\left(J_{AB}^{\mathcal{P}}\right)^{T}\otimes I_{A'B'}\right)\right]\\
 & =\mathrm{Tr}_{AB}\left[\mathbf{J}_{ABA'_{0}B'_{0}}^{\Theta}\left(\left(J_{AB}^{\mathcal{P}}\right)^{T}\otimes I_{A'_{0}B'_{0}}\right)\right]\\
 & =\mathrm{Tr}_{AB}\left[\left(\mathbf{J}_{A_{0}B_{0}A'_{0}B'_{0}}^{\Theta}\otimes u_{A_{1}B_{1}}\right)\left(\left(J_{AB}^{\mathcal{P}}\right)^{T}\otimes I_{A'_{0}B'_{0}}\right)\right]\\
 & =\frac{1}{\left|A_{1}B_{1}\right|}\mathrm{Tr}_{A_{0}B_{0}}\left[\mathbf{J}_{A_{0}B_{0}A'_{0}B'_{0}}^{\Theta}\left(\left(J_{A_{0}B_{0}}^{\mathcal{P}}\right)^{T}\otimes I_{A'_{0}B'_{0}}\right)\right]\\
 & =:\mathcal{D}_{A_{0}B_{0}\to A_{0}'B_{0}'}\left(J_{A_{0}B_{0}}^{\mathcal{P}}\right),
\end{align*}
where $\mathcal{D}$ is a CP map whose Choi matrix is given by $J_{A_{0}B_{0}A'_{0}B'_{0}}^{\mathcal{D}}:=\frac{1}{\left|A_{1}B_{1}\right|}\mathbf{J}_{A_{0}B_{0}A'_{0}B'_{0}}^{\Theta}$.
The fact that $\Theta$ is a superchannel ensures that $J_{A'_{0}B'_{0}}^{\mathcal{D}}=I_{A_{0}'B_{0}'}$,
so $\mathcal{D}$ is \emph{unital}. Now, the operator norm is contractive
under CP unital maps, thus we conclude that $\left\Vert J_{A_{0}'B_{0}'}^{\Theta\left[\mathcal{P}\right]}\right\Vert _{\infty}\leq\left\Vert J_{A_{0}B_{0}}^{\mathcal{P}}\right\Vert _{\infty}$.
Similarly, since $\Theta^{\Gamma}$ is also a superchannel, we have
$\left\Vert J_{A_{0}'B_{0}'}^{\Theta^{\Gamma}\left[\mathcal{P}^{\Gamma}\right]}\right\Vert _{\infty}\leq\left\Vert J_{A_{0}B_{0}}^{\mathcal{P}^{\Gamma}}\right\Vert _{\infty}$.
Therefore, recalling Eq.~\eqref{eq:ineq1},
\begin{align*}
 & LN_{\max}\left(\Theta\left[\mathcal{N}_{AB}\right]\right)\\
 & \leq\log_{2}\inf\left\{ \max\left\{ \left\Vert J_{A_{0}B_{0}}^{\mathcal{P}}\right\Vert _{\infty},\left\Vert J_{A_{0}B_{0}}^{\mathcal{P}^{\Gamma}}\right\Vert _{\infty}\right\} \right\} ,
\end{align*}
subject to $-\mathcal{P}_{AB}^{\Gamma}\leq\mathcal{N}_{AB}^{\Gamma}\leq\mathcal{P}_{AB}^{\Gamma}$
and $\mathcal{P}_{AB}\geq0$, where we have used the fact that $\Theta^{\Gamma}$
is a superchannel, so $-\mathcal{P}_{AB}^{\Gamma}\leq\mathcal{N}_{AB}^{\Gamma}\leq\mathcal{P}_{AB}^{\Gamma}$
implies $-\Theta^{\Gamma}\left[\mathcal{P}_{AB}^{\Gamma}\right]\leq\Theta^{\Gamma}\left[\mathcal{N}_{AB}^{\Gamma}\right]\leq\Theta^{\Gamma}\left[\mathcal{P}_{AB}^{\Gamma}\right]$.
But the final expression we have obtained is precisely $LN_{\max}\left(\mathcal{N}_{AB}\right)$.
This completes the proof.
\end{proof}

\paragraph{Additivity}

\noindent For any two bipartite channels $\mathcal{N}\in\mathrm{CPTP}\left(A_{0}B_{0}\to A_{1}B_{1}\right)$
and $\mathcal{M}\in\mathrm{CPTP}\left(A_{0}'B_{0}'\to A_{1}'B_{1}'\right)$
we have
\[
LN_{\max}\left(\mathcal{N}_{AB}\otimes\mathcal{M}_{A'B'}\right)=LN_{\max}\left(\mathcal{N}_{AB}\right)+LN_{\max}\left(\mathcal{M}_{A'B'}\right).
\]
In particular, note that this property implies that, for all positive
integers $n$, 
\[
LN_{\max}\left(\mathcal{N}_{AB}^{\otimes n}\right)=nLN_{\max}\left(\mathcal{N}_{AB}\right).
\]

\begin{proof}
\noindent The proof follows from the facts
\[
LN_{\max}^{\left(0\right)}\left(\mathcal{N}_{AB}\otimes\mathcal{M}_{A'B'}\right)=LN_{\max}^{\left(0\right)}\left(\mathcal{N}_{AB}\right)+LN_{\max}^{\left(0\right)}\left(\mathcal{M}_{A'B'}\right)
\]
and
\[
LN_{\max}^{\left(1\right)}\left(\mathcal{N}_{AB}\otimes\mathcal{M}_{A'B'}\right)=LN_{\max}^{\left(1\right)}\left(\mathcal{N}_{AB}\right)+LN_{\max}^{\left(1\right)}\left(\mathcal{M}_{A'B'}\right),
\]
which can be proved with the same techniques as in Ref.~\citep{WW18},
with the primal problem being used to show the $\leq$ side, and the
dual problem used to show the $\geq$ side. For completeness, we include
the proof in appendix~\ref{app:MLN}.
\end{proof}

\subsection{Exact asymptotic NPT entanglement cost\label{subsec:NPT cost}}

In this section we generalize the operational interpretation given
in Ref.~\citep{WW18} of $E_{\kappa}$ to\emph{ generic} bipartite
channels. This generalization will be fairly straightforward, and
the ultimate reason for this is that we do not consider only restricted
PPT superchannels, but rather generic PPT superchannels (see section~\ref{subsec:PPT-supermaps}).
This makes the conditions involved closer to the case of bipartite
states.

Following the same argument in section~\ref{sec:entanglement}, in
NPT entanglement theory, the maximally entangled state $\phi_{A_{1}'B_{1}'}^{+}$,
if suitably normalized, where $\left|A_{1}'\right|=\left|B_{1}'\right|=m$,
can be regarded as the maximal resource: two maximally entangled states
$\phi^{+}$ are equivalent to the swap. This state can also be viewed
as the channel $\Phi_{A'B'}^{+}$ (cf.~section~\ref{subsec:Supermaps})
with trivial inputs $A'_{0}$ and $B'_{0}$. With this in mind, the
single-shot exact resource cost to simulate a channel takes the form
\begin{equation}
E_{\mathrm{PPT}}^{\left(1\right)}\left(\mathcal{N}_{AB}\right):=\inf\left\{ \log_{2}m:\mathcal{N}_{AB}=\Theta_{A'B'\to AB}\left[\Phi_{A'B'}^{+}\right]\right\} ,\label{ee1}
\end{equation}
where the infimum is over all PPT superchannels $\Theta$, $\left|A_{0}'\right|=\left|B_{0}'\right|=1$,
and $\left|A_{1}'\right|=\left|B_{1}'\right|=m$.

The following two lemmas will be used in the proof of the main theorem
of this section (theorem~\ref{thm:oi}) that provides an operational
meaning to the MLN. The first lemma provides an alternative expression
for $E_{\mathrm{PPT}}^{\left(1\right)}\left(\mathcal{N}_{AB}\right)$. 
\begin{lem}
Let $\mathcal{N}\in\mathrm{CPTP}\left(A_{0}B_{0}\to A_{1}B_{1}\right)$
be a bipartite channel. Then,
\begin{align}
 & E_{\mathrm{PPT}}^{\left(1\right)}\left(\mathcal{N}_{AB}\right)\nonumber \\
 & =\inf\left\{ \log_{2}m:-\left(m-1\right)\mathcal{R}_{AB}^{\Gamma}\leq\mathcal{N}_{AB}^{\Gamma}\leq\left(m+1\right)\mathcal{R}_{AB}^{\Gamma}\right\} ,\label{eq:bb}
\end{align}
where \textup{$\mathcal{R}\in\mathrm{CPTP}\left(A_{0}B_{0}\to A_{1}B_{1}\right)$
and $m\in\mathbb{Z}_{+}$.}
\end{lem}

\begin{proof}
The proof follows similar lines to the one in Ref.~\citep{WW18},
but with states replaced by channels. We first prove that $E_{\mathrm{PPT}}^{\left(1\right)}\left(\mathcal{N}_{AB}\right)$
is less than or equal to the right-hand side of Eq.~\eqref{eq:bb}.
Let $m=\left|A'_{1}\right|=\left|B_{1}'\right|$ be a positive integer,
and let $\mathcal{R}_{AB}$ be a CPTP map satisfying 
\begin{equation}
-\left(m-1\right)\mathcal{R}_{AB}^{\Gamma}\leq\mathcal{N}_{AB}^{\Gamma}\leq\left(m+1\right)\mathcal{R}_{AB}^{\Gamma}.\label{ww}
\end{equation}
We need to show that there exists a PPT superchannel $\Theta$ as
in Eq.~\eqref{ee1} with the same $m$. To this end, define the superchannel
$\Theta\in\mathfrak{S}\left(A'B'\to AB\right)$ (with $\left|A_{1}'\right|=\left|B_{1}'\right|=m$
and $\left|A_{0}'\right|=\left|B_{0}'\right|=1$) on any CP map $\mathcal{M}_{A'B'}$
as
\begin{align*}
 & \Theta_{A'B'\to AB}\left[\mathcal{M}_{A'B'}\right]\\
 & :=\mathcal{N}_{AB}\mathrm{Tr}\left[\Phi_{A'B'}^{+}\mathcal{M}_{A'B'}\right]\\
 & +\mathcal{R}_{AB}\mathrm{Tr}\left[\left(I_{A'B'}-\Phi_{A'B'}^{+}\right)\mathcal{M}_{A'B'}\right],
\end{align*}
where we have used the fact that $\mathcal{M}_{A'B'}$ and $\Phi_{A'B'}^{+}$
can be viewed as matrices because their input dimensions are trivial,
so the traces above are well defined. For a simpler notation, set
$A'\equiv A'_{1}$ and $B'\equiv B'_{1}$.

Note that $\Theta$ above is indeed a superchannel, as it is CP, and
sends channels to channels \citep{Supermeasurements}. In addition,
it satisfies $\Theta\left[\Phi_{A'B'}^{+}\right]=\mathcal{N}_{AB}$.
We need to show that $\Theta^{\Gamma}=\Upsilon_{B}\circ\Theta\circ\Upsilon_{B'}$
is a superchannel too. For this purpose, let $R=\left(R_{0},R_{1}\right)$
be a reference system, and consider $\mathcal{P}_{RA'B'}\in\mathrm{CPTP}\left(R_{0}\to R_{1}A'B'\right)$,
and observe that
\begin{align*}
 & \Theta^{\Gamma}\left[\mathcal{P}_{RA'B'}\right]\\
 & =\mathcal{N}_{AB}^{\Gamma}\otimes\mathrm{Tr}_{A'B'}\left[\left(\phi_{A'B'}^{+}\right)^{T_{B'}}\mathcal{P}_{RA'B'}\right]\\
 & +\mathcal{R}_{AB}^{\Gamma}\otimes\mathrm{Tr}_{A'B'}\left[\left(I_{A'B'}-\phi_{A'B'}^{+}\right)^{T_{B'}}\mathcal{P}_{RA'B'}\right],
\end{align*}
where the partial trace above is understood as follows: for any matrix
$X\in\mathfrak{B}\left(R_{0}\right)$, the expression $\mathrm{Tr}_{A'B'}\left[\left(\phi_{A'B'}^{+}\right)^{T_{B'}}\mathcal{P}_{RA'B'}\right]$
is the map
\begin{align*}
 & \mathrm{Tr}_{A'B'}\left[\left(\phi_{A'B'}^{+}\right)^{T_{B'}}\mathcal{P}_{RA'B'}\right]\left(X_{R_{0}}\right)\\
 & :=\mathrm{Tr}_{A'B'}\left[\left(\phi_{A'B'}^{+}\right)^{T_{B'}}\mathcal{P}_{RA'B'}\left(X_{R_{0}}\right)\right].
\end{align*}

Recall that $\left(\phi_{A'B'}^{+}\right)^{T_{B'}}=\frac{1}{m}F_{A'B'}$,
where $F_{A'B'}$ is the unitary \noun{swap} (or flip) operator, and
the factor $\frac{1}{m}$ comes from the fact that here we are taking
$\phi_{A'B'}^{+}$ to be normalized. Therefore 
\[
\mathrm{Tr}_{A'B'}\left[\left(\phi_{A'B'}^{+}\right)^{T_{B'}}\mathcal{P}_{RA'B'}\right]=\frac{1}{m}\mathrm{Tr}_{A'B'}\left[F_{A'B'}\mathcal{P}_{RA'B'}\right],
\]
and
\begin{align*}
 & \mathrm{Tr}_{A'B'}\left[\left(I_{A'B'}-\phi_{A'B'}^{+}\right)^{T_{B'}}\mathcal{P}_{RA'B'}\right]\\
 & =\mathrm{Tr}_{A'B'}\left[\left(I_{A'B'}-\frac{1}{m}F_{A'B'}\right)\mathcal{P}_{RA'B'}\right].
\end{align*}
Following Ref.~\citep{WW18}, we define $\Pi_{A'B'}^{\pm}:=\frac{1}{2}\left(I_{A'B'}\pm F_{A'B'}\right)$
to be the orthogonal projections onto the symmetric and antisymmetric
subspaces respectively. Hence, substituting $\Pi_{A'B'}^{+}-\Pi_{A'B'}^{-}$
for $F_{A'B'}$, and $\Pi_{A'B'}^{+}+\Pi_{A'B'}^{-}$ for $I_{A'B'}$,
yields (cf.\ Eqs.~(68--73) in Ref.~\citep{WW18})
\begin{align}
 & \Theta^{\Gamma}\left[\mathcal{P}_{RA'B'}\right]\nonumber \\
 & =\frac{1}{m}\mathcal{N}_{AB}^{\Gamma}\otimes\mathrm{Tr}_{A'B'}\left[F_{A'B'}\mathcal{P}_{RA'B'}\right]\nonumber \\
 & +\mathcal{R}_{AB}^{\Gamma}\otimes\mathrm{Tr}_{A'B'}\left[\left(I_{A'B'}-\frac{1}{m}F_{A'B'}\right)\mathcal{P}_{RA'B'}\right]\nonumber \\
 & =\frac{1}{m}\left(\mathcal{N}_{AB}^{\Gamma}+\left(m-1\right)\mathcal{R}_{AB}^{\Gamma}\right)\otimes\mathrm{Tr}_{A'B'}\left[\Pi_{A'B'}^{+}\mathcal{P}_{RA'B'}\right]\nonumber \\
 & +\frac{1}{m}\left(\left(m+1\right)\mathcal{R}_{AB}^{\Gamma}-\mathcal{N}_{AB}^{\Gamma}\right)\otimes\mathrm{Tr}_{A'B'}\left[\Pi_{A'B'}^{-}\mathcal{P}_{RA'B'}\right].\label{eq:pptt}
\end{align}
By Eq.~\eqref{ww}, the expression on the right-hand side of the
equation above is a CPTP map. Hence, $\mathbbm{1}_{R}\otimes\Theta^{\Gamma}$
takes channels to channels; i.e.\ $\Theta^{\Gamma}$ is a superchannel,
so $\Theta$ is indeed a PPT superchannel. To summarize, we showed
that, for any integer $m$ for which there exists a channel $\mathcal{R}_{AB}$
that satisfies Eq.~\eqref{ww}, there exists a PPT superchannel $\Theta$
achieving $\Theta\left[\Phi_{A'B'}^{+}\right]=\mathcal{N}_{AB}$ with
$\left|A'_{1}\right|=\left|B'_{1}\right|=m$ (and $\left|A'_{0}\right|=\left|B'_{0}\right|=1$).
Hence, $E_{\mathrm{PPT}}^{\left(1\right)}\left(\mathcal{N}_{AB}\right)$
cannot be greater than the right-hand side of Eq.~\eqref{eq:bb}.
To complete the proof, we now prove the converse inequality; i.e.\ we
show that $E_{\mathrm{PPT}}^{\left(1\right)}\left(\mathcal{N}_{AB}\right)$
is greater than or equal to the right-hand side of Eq.~\eqref{eq:bb}.
Denote by $\mathcal{G}\in\mathrm{CPTP}\left(A'B'\to A'B'\right)$
the twirling channel of the form 
\begin{equation}
\mathcal{G}\left(\omega_{A'B'}\right)=\int\left(U_{A'}\otimes\overline{U}_{B'}\right)\omega_{A'B'}\left(U_{A'}\otimes\overline{U}_{B'}\right)^{\dagger}dU_{A'B'}\label{twirl}
\end{equation}
with respect to the Haar probability measure, $dU$, over unitary
matrices. It can be shown \citep{Werner-states,Horodecki-reduction,Watrous}
that $\mathcal{G}$ is actually the channel
\begin{align*}
\mathcal{G}\left(\omega_{A'B'}\right) & =\phi_{A'B'}^{+}\mathrm{Tr}\left[\phi_{A'B'}^{+}\omega_{A'B'}\right]\\
 & +\frac{I_{A'B'}-\phi_{A'B'}^{+}}{m^{2}-1}\mathrm{Tr}\left[\left(I_{A'B'}-\phi_{A'B'}^{+}\right)\omega_{A'B'}\right].
\end{align*}
Note that, since $\left|A'_{0}\right|=\left|B_{0}'\right|=1$, we
can view the channel $\mathcal{G}$ as a superchannel $\mathcal{G}_{A'B'\to A'B'}$
taking channels (which are nothing but density matrices) in $\mathrm{CPTP}\left(A_{0}'B_{0}'\to A_{1}'B_{1}'\right)$
to channels in the same set. In particular, this superchannel is self-adjoint,
and satisfies $\mathcal{G}\left[\Phi_{A'B'}^{+}\right]=\Phi_{A'B'}^{+}$.
The latter implies that if $\Theta$ is a PPT superchannel such that
$\Theta\left[\Phi_{A'B'}^{+}\right]=\mathcal{N}_{AB}$, then $\Omega_{A'B'\to AB}:=\Theta_{A'B'\to AB}\circ\mathcal{G}_{A'B'}$
is also a PPT superchannel that takes $\Phi_{A'B'}^{+}$ to $\mathcal{N}_{AB}$
(i.e.\ it achieves the same performance as $\Theta$). Furthermore,
by Eq.~\eqref{twirl} the superchannel $\Omega$ satisfies (cf.\ Eqs.~(80)--(82)
of Ref.~\citep{WW18})
\begin{align*}
\Omega_{A'B'\to AB}\left[\mathcal{M}_{A'B'}\right] & :=\mathcal{N}_{AB}\mathrm{Tr}\left[\Phi_{A'B'}^{+}\mathcal{M}_{A'B'}\right]\\
 & +\mathcal{R}_{AB}\mathrm{Tr}\left[\left(I_{A'B'}-\Phi_{A'B'}^{+}\right)\mathcal{M}_{A'B'}\right],
\end{align*}
where
\[
\mathcal{R}_{AB}:=\frac{1}{1-m^{2}}\Theta\left[I_{A'B'}-\Phi_{A'B'}^{+}\right].
\]
Now, from the exact same lines leading to Eq.~\eqref{eq:pptt}, it
follows that, for $\Omega_{A'B'\to AB}$ to be a PPT superchannel,
it is necessary that for any $\mathcal{P}_{RA'B'}\in\mathrm{CPTP}\left(R_{0}\to R_{1}A_{1}'B_{1}'\right)$,
the map on the right-hand side of Eq.~\eqref{eq:pptt} is a quantum
channel. Since $\Pi^{+}$ and $\Pi^{-}$ are orthogonal projectors,
each term must be a CP map, which yields Eq.~\eqref{ww}. To summarize,
if $\Theta$ is a PPT superchannel that satisfies $\Theta\left[\Phi_{A'B'}^{+}\right]=\mathcal{N}_{AB}$,
then $\Omega$ is also a PPT superchannel that satisfies $\Omega\left[\Phi_{A'B'}^{+}\right]=\mathcal{N}_{AB}$;
the fact that $\Omega$ is PPT forces each term of Eq.~\eqref{eq:pptt}
to be a CP map which is equivalent to Eq.~\eqref{ww}. Hence, $E_{\mathrm{PPT}}^{\left(1\right)}\left(\mathcal{N}_{AB}\right)$
cannot be smaller than the right-hand side of Eq.\textbf{~}\eqref{eq:bb}.
This completes the proof.
\end{proof}
The second lemma uses the previous one to link the the single-shot
exact PPT cost to the MLN.
\begin{lem}
\label{lem:double-bound}Let $\mathcal{N}\in\mathrm{CPTP}\left(A_{0}B_{0}\to A_{1}B_{1}\right)$
be a bipartite channel. Then,
\begin{align*}
\log_{2}\left(2^{LN_{\max}\left(\mathcal{N}_{AB}\right)}-1\right) & \leq E_{\mathrm{PPT}}^{\left(1\right)}\left(\mathcal{N}_{AB}\right)\\
 & \leq\log_{2}\left(2^{LN_{\max}\left(\mathcal{N}_{AB}\right)}+2\right).
\end{align*}
\end{lem}

\begin{proof}
First of all, we prove that the result of the previous lemma can be
rewritten in a slightly modified version:
\begin{align}
 & E_{\mathrm{PPT}}^{\left(1\right)}\left(\mathcal{N}_{AB}\right)\nonumber \\
 & =\inf\left\{ \log_{2}m:-\left(m-1\right)\mathcal{R}_{AB}^{\Gamma}\leq\mathcal{N}_{AB}^{\Gamma}\leq\left(m+1\right)\mathcal{R}_{AB}^{\Gamma}\right\} ,\label{eq:hh2}
\end{align}
where $\mathcal{R}\geq0$, $J_{A_{0}B_{0}}^{\mathcal{R}}\leq I_{A_{0}B_{0}}$,
$J_{A_{0}B_{0}}^{\mathcal{R}^{\Gamma}}\leq I_{A_{0}B_{0}}$, and $m\in\mathbb{N}$.
To see why, denote the second line of Eq.~\eqref{eq:hh2} by $\widetilde{E}_{\mathrm{PPT}}^{\left(1\right)}\left(\mathcal{N}_{AB}\right)$.
Then, by definition, we have $E_{\mathrm{PPT}}^{\left(1\right)}\left(\mathcal{N}_{AB}\right)\leq\widetilde{E}_{\mathrm{PPT}}^{\left(1\right)}\left(\mathcal{N}_{AB}\right)$
because if $\mathcal{R}$ is a CPTP, then $J_{A_{0}B_{0}}^{\mathcal{R}}=J_{A_{0}B_{0}}^{\mathcal{R}^{\Gamma}}=I_{A_{0}B_{0}}$
(note that the condition $-\left(m-1\right)\mathcal{R}_{AB}^{\Gamma}\leq\left(m+1\right)\mathcal{R}_{AB}^{\Gamma}$
implies in particular that $\mathcal{R}^{\Gamma}\geq0$). Conversely,
suppose $\mathcal{R}$ satisfies $J_{A_{0}B_{0}}^{\mathcal{R}}\leq I_{A_{0}B_{0}}$
and $J_{A_{0}B_{0}}^{\mathcal{R}^{\Gamma}}\leq I_{A_{0}B_{0}}$. Define
$\mathcal{P}$ as the map whose Choi matrix is given by 
\[
J_{AB}^{\mathcal{P}}:=J_{AB}^{\mathcal{R}}+\left(I_{A_{0}B_{0}}-J_{A_{0}B_{0}}^{\mathcal{R}}\right)\otimes u_{A_{1}B_{1}}.
\]
Note that $\mathcal{P}$ is a channel, and that both $\left(I_{A_{0}B_{0}}-J_{A_{0}B_{0}}^{\mathcal{R}}\right)\otimes u_{A_{1}B_{1}}$
and its partial transpose are positive semi-definite. Therefore, $\mathcal{P}$
too satisfies the constraints 
\[
-\left(m-1\right)\mathcal{P}_{AB}^{\Gamma}\leq\mathcal{N}_{AB}^{\Gamma}\leq\left(m+1\right)\mathcal{P}_{AB}^{\Gamma},
\]
so we can conclude that $E_{\mathrm{PPT}}^{\left(1\right)}\left(\mathcal{N}_{AB}\right)\geq\widetilde{E}_{\mathrm{PPT}}^{\left(1\right)}\left(\mathcal{N}_{AB}\right)$.
This proves that $E_{\mathrm{PPT}}^{\left(1\right)}\left(\mathcal{N}_{AB}\right)=\widetilde{E}_{\mathrm{PPT}}^{\left(1\right)}\left(\mathcal{N}_{AB}\right)$.

The rest of the proof employs similar techniques to proposition~9
in Ref.~\citep{WW18}, with a few exceptions. Continuing, we have
\begin{widetext}
\begin{align}
 & E_{\mathrm{PPT}}^{\left(1\right)}\left(\mathcal{N}_{AB}\right)=\widetilde{E}_{\mathrm{PPT}}^{\left(1\right)}\left(\mathcal{N}_{AB}\right)\nonumber \\
 & \geq\log_{2}\inf\left\{ m:-\left(m+1\right)\mathcal{R}_{AB}^{\Gamma}\leq\mathcal{N}_{AB}^{\Gamma}\leq\left(m+1\right)\mathcal{R}_{AB}^{\Gamma},\,\mathcal{R}\geq0,\,J_{A_{0}B_{0}}^{\mathcal{R}}\leq I_{A_{0}B_{0}},\thinspace J_{A_{0}B_{0}}^{\mathcal{R}^{\Gamma}}\leq I_{A_{0}B_{0}},\thinspace m\in\mathbb{N}\right\} \label{x1}\\
 & =\log_{2}\inf\left\{ m:-\mathcal{P}_{AB}^{\Gamma}\leq\mathcal{N}_{AB}^{\Gamma}\leq\mathcal{P}_{AB}^{\Gamma},\thinspace\mathcal{P}\geq0,\thinspace J_{A_{0}B_{0}}^{\mathcal{P}}\leq\left(m+1\right)I_{A_{0}B_{0}},\thinspace J_{A_{0}B_{0}}^{\mathcal{P}^{\Gamma}}\leq\left(m+1\right)I_{A_{0}B_{0}},\thinspace m\in\mathbb{N}\right\} \label{x2}\\
 & \geq\log_{2}\inf\left\{ m:-\mathcal{P}_{AB}^{\Gamma}\leq\mathcal{N}_{AB}^{\Gamma}\leq\mathcal{P}_{AB}^{\Gamma},\thinspace\mathcal{P}\geq0,\thinspace J_{A_{0}B_{0}}^{\mathcal{P}}\leq\left(m+1\right)I_{A_{0}B_{0}},\thinspace J_{A_{0}B_{0}}^{\mathcal{P}^{\Gamma}}\leq\left(m+1\right)I_{A_{0}B_{0}},\thinspace m\in\mathbb{R}_{+}\right\} \label{x3}\\
 & \geq\log_{2}\inf\left\{ \max\left\{ \left\Vert J_{A_{0}B_{0}}^{\mathcal{P}}\right\Vert _{\infty},\left\Vert J_{A_{0}B_{0}}^{\mathcal{P}^{\Gamma}}\right\Vert _{\infty}\right\} -1:-\mathcal{P}_{AB}^{\Gamma}\leq\mathcal{N}_{AB}^{\Gamma}\leq\mathcal{P}_{AB}^{\Gamma},\thinspace\mathcal{P}\geq0\right\} \label{x4}\\
 & =\log_{2}\left(2^{LN_{\max}\left(\mathcal{N}_{AB}\right)}-1\right),\nonumber 
\end{align}
\end{widetext}

\noindent where in Eq.~\eqref{x1} we replaced $m-1$ with $m+1$,
so the infimum is on a less restricted set, in Eq.~\eqref{x2} we
defined $\mathcal{P}_{AB}:=\left(m+1\right)\mathcal{R}_{AB}$, in
Eq.~\eqref{x3} we removed the restriction that $m$ is an integer,
and the last inequality~\eqref{x4} follows from the fact that if
$J_{A_{0}B_{0}}^{\mathcal{P}}\leq\left(m+1\right)I_{A_{0}B_{0}}$
and $J_{A_{0}B_{0}}^{\mathcal{P}^{\Gamma}}\leq\left(m+1\right)I_{A_{0}B_{0}}$
then $m\geq\max\left\{ \left\Vert J_{A_{0}B_{0}}^{\mathcal{P}}\right\Vert _{\infty},\left\Vert J_{A_{0}B_{0}}^{\mathcal{P}^{\Gamma}}\right\Vert _{\infty}\right\} -1$.

For the other inequality, following similar lines, we get
\begin{widetext}
\begin{align}
 & 2^{E_{\mathrm{PPT}}^{\left(1\right)}\left(\mathcal{N}_{AB}\right)}=2^{\widetilde{E}_{\mathrm{PPT}}^{\left(1\right)}\left(\mathcal{N}_{AB}\right)}\label{eq:m-1}\\
 & \leq\inf\left\{ m:-\left(m-1\right)\mathcal{R}_{AB}^{\Gamma}\leq\mathcal{N}_{AB}^{\Gamma}\leq\left(m-1\right)\mathcal{R}_{AB}^{\Gamma},\thinspace\mathcal{R}\geq0,\thinspace J_{A_{0}B_{0}}^{\mathcal{R}}\leq I_{A_{0}B_{0}},\thinspace J_{A_{0}B_{0}}^{\mathcal{R}^{\Gamma}}\leq I_{A_{0}B_{0}},\thinspace m\in\mathbb{N}\right\} \label{eq:P}\\
 & =\inf\left\{ m:-\mathcal{P}_{AB}^{\Gamma}\leq\mathcal{N}_{AB}^{\Gamma}\leq\mathcal{P}_{AB}^{\Gamma},\thinspace\mathcal{P}\geq0,\thinspace J_{A_{0}B_{0}}^{\mathcal{P}}\leq\left(m-1\right)I_{A_{0}B_{0}},\thinspace J_{A_{0}B_{0}}^{\mathcal{P}^{\Gamma}}\leq\left(m-1\right)I_{A_{0}B_{0}},\thinspace m\in\mathbb{N}\right\} \nonumber \\
 & =\inf\left\{ \left\lfloor m\right\rfloor :-\mathcal{P}_{AB}^{\Gamma}\leq\mathcal{N}_{AB}^{\Gamma}\leq\mathcal{P}_{AB}^{\Gamma},\thinspace\mathcal{P}\geq0,\thinspace J_{A_{0}B_{0}}^{\mathcal{P}}\leq\left(\left\lfloor m\right\rfloor -1\right)I_{A_{0}B_{0}},\thinspace J_{A_{0}B_{0}}^{\mathcal{P}^{\Gamma}}\leq\left(\left\lfloor m\right\rfloor -1\right)I_{A_{0}B_{0}},\thinspace m\in\mathbb{R}_{+}\right\} \label{y0}\\
 & \leq\inf\left\{ m:-\mathcal{P}_{AB}^{\Gamma}\leq\mathcal{N}_{AB}^{\Gamma}\leq\mathcal{P}_{AB}^{\Gamma},\thinspace\mathcal{P}\geq0,\thinspace J_{A_{0}B_{0}}^{\mathcal{P}}\leq\left(m-2\right)I_{A_{0}B_{0}},\thinspace J_{A_{0}B_{0}}^{\mathcal{P}^{\Gamma}}\leq\left(m-2\right)I_{A_{0}B_{0}},\thinspace m\in\mathbb{R}_{+}\right\} \label{y1}\\
 & =\inf\left\{ \max\left\{ \left\Vert J_{A_{0}B_{0}}^{\mathcal{P}}\right\Vert _{\infty},\left\Vert J_{A_{0}B_{0}}^{\mathcal{P}^{\Gamma}}\right\Vert _{\infty}\right\} +2:-\mathcal{P}_{AB}^{\Gamma}\leq\mathcal{N}_{AB}^{\Gamma}\leq\mathcal{P}_{AB}^{\Gamma},\thinspace\mathcal{P}\geq0\right\} \nonumber \\
 & =2^{LN_{\max}\left(\mathcal{N}_{AB}\right)}+2,\nonumber 
\end{align}
\end{widetext}

\noindent where in Eq.~\eqref{eq:m-1} we replaced $m+1$ with $m\in\mathbb{R}_{+}$,
obtaining a larger set, in Eq.~\eqref{eq:P} we set $\mathcal{P}:=\left(m-1\right)\mathcal{R}$,
in Eq.~\eqref{y1} we used the fact that $m-2\leq\left\lfloor m\right\rfloor -1$
so, the constraints $J_{A_{0}B_{0}}^{\mathcal{P}}\leq\left(m-2\right)I_{A_{0}B_{0}}$
and $J_{A_{0}B_{0}}^{\mathcal{P}^{\Gamma}}\leq\left(m-2\right)I_{A_{0}B_{0}}$
imply the constraints $J_{A_{0}B_{0}}^{\mathcal{P}}\leq\left(\left\lfloor m\right\rfloor -1\right)I_{A_{0}B_{0}}$
and $J_{A_{0}B_{0}}^{\mathcal{P}^{\Gamma}}\leq\left(\left\lfloor m\right\rfloor -1\right)I_{A_{0}B_{0}}$
of Eq.~\eqref{y0} respectively. This completes the proof.
\end{proof}
Recalling Ref.~\citep{Gour-Scandolo-resource}, the exact (parallel)
NPT entanglement cost of the channel is defined as 
\[
E_{\mathrm{PPT}}\left(\mathcal{N}_{AB}\right)=\limsup_{n}\frac{1}{n}E_{\mathrm{PPT}}^{\left(1\right)}\left(\mathcal{N}_{AB}^{\otimes n}\right).
\]
The following result, which is the key theorem of this section, states
that the exact PPT cost of a bipartite channel is given precisely
by its max-logarithmic negativity.
\begin{thm}
\label{thm:oi}Let $\mathcal{N}\in\mathrm{CPTP}\left(A_{0}B_{0}\to A_{1}B_{1}\right)$
be a bipartite channel. Then, 
\[
E_{\mathrm{PPT}}\left(\mathcal{N}_{AB}\right)=LN_{\max}\left(\mathcal{N}_{AB}\right).
\]
 
\end{thm}

\begin{proof}
The proof follows from the additivity property of $LN_{\max}\left(\mathcal{N}_{AB}\right)$
and lemma~\ref{lem:double-bound}. Specifically, 
\begin{align*}
\frac{1}{n}E_{\mathrm{PPT}}^{\left(1\right)}\left(\mathcal{N}_{AB}^{\otimes n}\right) & \leq\frac{1}{n}\log_{2}\left(2^{LN_{\max}\left(\mathcal{N}_{AB}^{\otimes n}\right)}+2\right)\\
 & =\frac{1}{n}\log_{2}\left(2^{nLN_{\max}\left(\mathcal{N}_{AB}\right)}+2\right).
\end{align*}
Conversely, 
\begin{align*}
\frac{1}{n}E_{\mathrm{PPT}}^{\left(1\right)}\left(\mathcal{N}_{AB}^{\otimes n}\right) & \geq\frac{1}{n}\log_{2}\left(2^{LN_{\max}\left(\mathcal{N}_{AB}^{\otimes n}\right)}-1\right)\\
 & =\frac{1}{n}\log_{2}\left(2^{nLN_{\max}\left(\mathcal{N}_{AB}\right)}-1\right).
\end{align*}
 Then 
\begin{align*}
\frac{1}{n}\log_{2}\left(2^{nLN_{\max}\left(\mathcal{N}_{AB}\right)}-1\right) & \leq\frac{1}{n}E_{\mathrm{PPT}}^{\left(1\right)}\left(\mathcal{N}_{AB}^{\otimes n}\right)\\
 & \leq\frac{1}{n}\log_{2}\left(2^{nLN_{\max}\left(\mathcal{N}_{AB}\right)}+2\right).
\end{align*}
If we take the limit as $n\rightarrow+\infty$, the lower and upper
bound of $\frac{1}{n}E_{\mathrm{PPT}}^{\left(1\right)}\left(\mathcal{N}_{AB}^{\otimes n}\right)$
have the the same limit, equal to $LN_{\max}\left(\mathcal{N}_{AB}\right)$.
Therefore $\frac{1}{n}E_{\mathrm{PPT}}^{\left(1\right)}\left(\mathcal{N}_{AB}^{\otimes n}\right)$
has the limit (which will be equal to its limit superior). This allows
us to conclude that $E_{\mathrm{PPT}}\left(\mathcal{N}_{AB}\right)=LN_{\max}\left(\mathcal{N}_{AB}\right)$.
\end{proof}
In Ref.~\citep{Dynamical-entanglement} we proved that the MLN is
an upper bound for another entanglement measure, the NPT entanglement
generation power $E_{g}^{\mathrm{PPT}}$ \citep{Bennett-bipartite,Resource-channels-1,Resource-channels-2,Gour-Winter}:
\[
E_{g}^{\mathrm{PPT}}\left(\mathcal{N}_{AB}\right)\leq LN_{\max}\left(\mathcal{N}_{AB}\right).
\]

\section{SEP entanglement of a bipartite channel\label{sec:SEP}}

In the previous section we saw that extending the set of free operations
beyond LOCC can be very fruitful. However, one may argue that the
PPT operations allow for ``too much'' freedom, making NPT entanglement
a rather crude approximation of LOCC-entanglement. Here we consider
a much smaller set: the set of separable superchannels (SEPS). Like
before, SEPS do not necessarily have a realization similar to the
one in Fig.~\ref{fig:PPT}, where the pre-processing and post-processing
are both SEP channels. Instead, we define SEPS using the Choi matrix
formalism of superchannels. This simplifies the set of operations,
making them more useful for applications and calculations.

Recall that a channel $\mathcal{N}\in\mathrm{CPTP}\left(A_{0}B_{0}\to A_{1}B_{1}\right)$
is called \emph{separable} \citep{SEP,Rains-SEP,PPT1} if it has an
operator-sum representation of the form 
\[
\mathcal{N}_{AB}\left(\rho_{A_{0}B_{0}}\right)=\sum_{j}\left(X_{A_{0}}^{j}\otimes Y_{B_{0}}^{j}\right)\rho_{A_{0}B_{0}}\left(X_{A_{0}}^{j}\otimes Y_{B_{0}}^{j}\right)^{\dagger},
\]
where $X^{j}\in\mathfrak{B}\left(A_{0}\right)$, $Y^{j}\in\mathfrak{B}\left(B_{0}\right)$,
and $\sum_{j}\left(X_{A_{0}}^{j}\right)^{\dagger}X_{A_{0}}^{j}\otimes\left(Y_{B_{0}}^{j}\right)^{\dagger}Y_{B_{0}}^{j}=I_{A_{0}B_{0}}$.
It is simple to check that the set SEP is precisely the set of completely
resource non-generating operations \citep{Quantum-resource-2,Gour-Scandolo-resource}
in entanglement theory (see e.g.\ Ref.~\citep{Chitambar-2018b}
and references therein). Moreover, a bipartite channel is separable
if and only if its Choi matrix is a separable matrix. This fact inspires
us to define SEPS using the Choi formalism for superchannels.
\begin{defn}
\label{def:SEPS}Let $\Theta\in\mathfrak{S}\left(A'B'\to A'B'\right)$
be a bipartite superchannel. Then, $\Theta$ is called a \emph{separable
superchannel (SEPS)} if its Choi matrix is separable; i.e.\ it can
be expressed as 
\[
\mathbf{J}_{ABA'B'}^{\Theta}=\sum_{j}X_{AA'}^{j}\otimes Y_{BB'}^{j},
\]
where, for all $j$, the matrices $X_{AA'}^{j}$ and $Y_{BB'}^{j}$
are positive semi-definite. We denote by ${\rm SEPS}\left(AB\to A'B'\right)$
the set of all bipartite SEPS from system $AB$ to $A'B'$.
\end{defn}

\begin{rem}
Note that clearly SEPS is a subset of PPT superchannels.
\end{rem}

Definition~\ref{def:SEPS} does \emph{not} refer to the implementation
of SEPS with pre- and post-processing that are both SEP channels.
On the other hand, however, if a bipartite superchannel $\Theta$
consists of a SEP pre-processing channel $\mathcal{E}$ and a SEP
post-processing channel $\mathcal{F}$, then the channel $\mathcal{Q}^{\Theta}=\mathcal{F}\circ\mathcal{E}$
is also SEP (and also its Choi matrix $\mathbf{J}^{\Theta}$), so
we can conclude that $\Theta$ is SEPS.

The next proposition shows that the set of SEPS is not ``too large'',
in the sense that it cannot generate (dynamical) entanglement out
of SEP channels. In this way, we establish that a superchannel $\Theta$
is completely non-entangling (i.e.\ completely resource non-generating)
if and only if it is a SEPS.
\begin{prop}
$\Theta\in{\rm SEPS}\left(AB\to A'B'\right)$ if and only if, for
every trace non-increasing separable CP map $\mathcal{N}_{A''AB''B}\in\mathrm{CP}\left(A''_{0}A_{0}B''_{0}B_{0}\to A''_{1}A_{1}B''_{1}B_{1}\right)$,
the map 
\[
\left(\mathbbm{1}_{A''B''}\otimes\Theta_{AB\to A'B'}\right)\left[\mathcal{N}_{A''AB''B}\right],
\]
is a separable trace non-increasing CP map in $\mathrm{CP}\left(A''_{0}A_{0}'B''_{0}B_{0}'\to A''_{1}A_{1}'B''_{1}B_{1}'\right)$.
\end{prop}

\begin{proof}
Let $\Theta$ be SEPS. Note that
\[
\mathbf{J}_{A''B''ABA'B'}^{\mathbbm{1}\otimes\Theta}=\mathbf{J}_{A''B''}^{\mathbbm{1}}\otimes\mathbf{J}_{ABA'B'}^{\Theta},
\]
where 
\begin{align*}
\mathbf{J}_{A''B''}^{\mathbbm{1}} & =\phi_{A_{0}''B_{0}''\widetilde{A}_{0}''\widetilde{B}_{0}''}^{+}\otimes\phi_{A_{1}''B_{1}''\widetilde{A}_{1}''\widetilde{B}_{1}''}^{+}\\
 & =\phi_{A_{0}''\widetilde{A}_{0}''}^{+}\otimes\phi_{A_{1}''\widetilde{A}_{1}''}^{+}\otimes\phi_{B_{0}''\widetilde{B}_{0}''}^{+}\otimes\phi_{B_{1}''\widetilde{B}_{1}''}^{+},
\end{align*}
is separable. Since $\mathbf{J}_{ABA'B'}^{\Theta}$ is also separable,
then $\mathbbm{1}_{A''B''}\otimes\Theta_{AB\to A'B'}$ is in SEPS
too. Hence, it is enough to show that $\Theta$ is RNG. Let $\mathcal{M}_{AB}$
be a separable bipartite CP map. Then, 
\[
J_{A'B'}^{\Theta\left[\mathcal{M}\right]}=\mathrm{Tr}_{AB}\left[\mathbf{J}_{ABA'B'}^{\Theta}\left(\left(J_{AB}^{\mathcal{M}}\right)^{T}\otimes I_{A'B'}\right)\right]
\]
is separable since both $\mathbf{J}_{ABA'B'}^{\Theta}$ and $J_{AB}^{\mathcal{M}}$
are separable.

Conversely, suppose $\Theta\in\mathfrak{S}\left(AB\to A'B'\right)$
is a completely non-entangling superchannel with respect to SEP channels.
Recall the representation of $\Theta$ given by $\mathcal{P}^{\Theta}$
as in section~\ref{subsec:Supermaps}, where $A$ and $B$ are replaced
by $AB$ and $A'B'$ respectively. We have 
\begin{align*}
\mathcal{P}_{ABA'B'}^{\Theta} & =\Theta_{\widetilde{A}\widetilde{B}\to A'B'}\left[\Phi_{AB\widetilde{A}\widetilde{B}}^{+}\right]\\
 & =\Theta_{\widetilde{A}\widetilde{B}\to A'B'}\left[\Phi_{A\widetilde{A}}^{+}\otimes\Phi_{B\widetilde{B}}^{+}\right],
\end{align*}
where we have used the fact that the CP map $\Phi_{ABA'B'}^{+}$ splits
in exactly the same way as its state counterpart $\phi_{ABA'B'}^{+}$.
Since $\Theta$ is completely non-entangling, it follows that the
channel $\mathcal{P}_{ABA'B'}^{\Theta}$ is separable, and therefore
its Choi matrix $\mathbf{J}_{ABA'B'}^{\Theta}$ is separable as well.
Hence, $\Theta$ is a SEPS. This completes the proof.
\end{proof}

\section{Bound dynamical entanglement\label{sec:Bound}}

We know that if the partial transpose of a bipartite entangled state
yields a positive semi-definite matrix, then the state is not distillable
under LOCC \citep{PHorodecki,Bound-entanglement}. Such states are
said to possess \emph{bound entanglement} \citep{Bound-entanglement}. 

This condition can be elevated to bipartite channels. Let $\mathcal{N}_{AB}\in\mathrm{CPTP}\left(A_{0}B_{0}\to A_{1}B_{1}\right)$
be a bipartite channel whose partial transpose $\mathcal{N}_{AB}^{\Gamma}$
is also a bipartite channel (i.e.\ $\mathcal{N}_{AB}$ is a PPT bipartite
channel). We argue here that such channels cannot be used to distill
entanglement. To see why, by contradiction, suppose that there exists
$n\in\mathbb{N}$ large enough and an LOCC superchannel $\Theta$
converting $\mathcal{N}_{AB}^{\otimes n}$ to a bipartite qubit state
$\rho_{A'B'}=\Theta\left[\mathcal{N}_{AB}^{\otimes n}\right]$, where
$\left|A_{0}'\right|=\left|B_{0}'\right|=1$ and $\left|A_{1}'\right|=\left|B'_{1}\right|=2$.
If $\rho_{A'B'}$ is entangled, its partial transpose is not positive
semi-definite \citep{PPT-Peres,PPT-Horodecki}. On the other hand,
on the right-hand side the partial transpose gives 
\[
\left(\Theta\left[\mathcal{N}_{AB}^{\otimes n}\right]\right)^{\Gamma}=\Theta^{\Gamma}\left[\left(\mathcal{N}_{AB}^{\otimes n}\right)^{\Gamma}\right]=\Theta^{\Gamma}\left[\left(\mathcal{N}_{AB}^{\Gamma}\right)^{\otimes n}\right]\geq0
\]
for LOCC superchannels are in particular PPT, so $\Theta^{\Gamma}$
is a superchannel. Recall also that we assume that $\mathcal{N}_{AB}^{\Gamma}$
is a channel as well. Therefore, we get a contradiction.

Note that in the argument above we showed that PPT superchannels (which
include in particular LOCC superchannels) cannot be used to distill
entanglement from an arbitrarily large number of copies of a PPT channel.
This further shows that our definition of the set of PPT superchannels,
which in principle can be larger than the set of superchannels realizable
with PPT pre- and post-processing as in Fig.~\ref{fig:PPT}, is not
so large such that PPT entanglement becomes distillable.

So far we have discussed the parallel scenario in which the superchannel
$\Theta$ acts on $\mathcal{N}_{AB}^{\otimes n}$ in parallel, or
at a single time. However, if one can use the channel repeatedly and
sequentially, one can realize e.g.\ a transformation of the form
\begin{equation}
\Theta_{n}\left[\mathcal{N}_{AB}\right]\circ\dots\circ\Theta_{2}\left[\mathcal{N}_{AB}\right]\circ\Theta_{1}\left[\mathcal{N}_{AB}\right],\label{sequenceppt}
\end{equation}
as illustrated in Fig.~\ref{fig:pptsequence}.
\begin{figure}
\begin{centering}
\includegraphics[width=1\columnwidth]{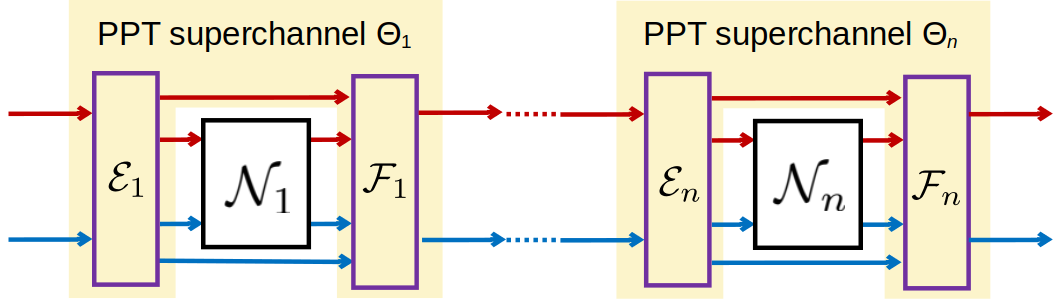}
\par\end{centering}
\caption{\label{fig:pptsequence}Sequence of PPT superchannels applied to the
channels $\mathcal{N}_{1},\dots,\mathcal{N}_{n}$.}

\end{figure}
 More generally, in Fig.~\ref{fig:pptcomb} we illustrate a PPT comb,
which is not necessarily of the form given in Eq.~\eqref{sequenceppt}.
\begin{figure}
\begin{centering}
\includegraphics[width=1\columnwidth]{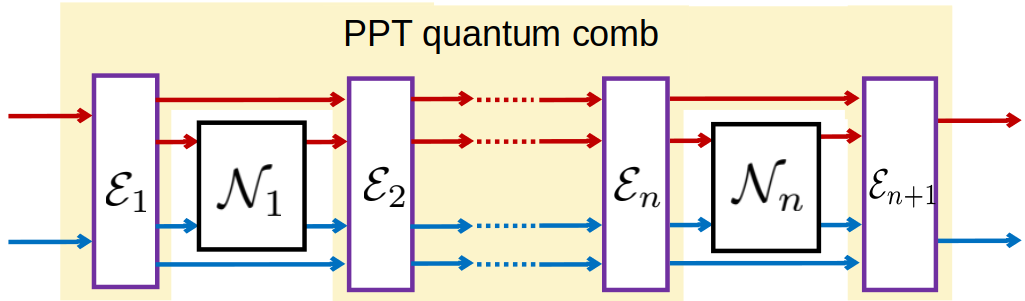}
\par\end{centering}
\caption{\label{fig:pptcomb}PPT comb: any comb such that $\mathcal{E}_{n+1}\circ\mathcal{E}_{n}\circ\dots\circ\mathcal{E}_{1}$
is a PPT channel.}

\end{figure}
 Instead, for a PPT comb we only require that the channel $\mathcal{Q}_{A^{n+1}B^{n+1}}:=\mathcal{E}_{n+1}\circ\mathcal{E}_{n}\circ\dots\circ\mathcal{E}_{1}$
be PPT. The channel $\mathcal{Q}_{A^{n+1}B^{n+1}}$ is illustrated
in Fig.~\ref{fig:pptq}.
\begin{figure}
\begin{centering}
\includegraphics[width=1\columnwidth]{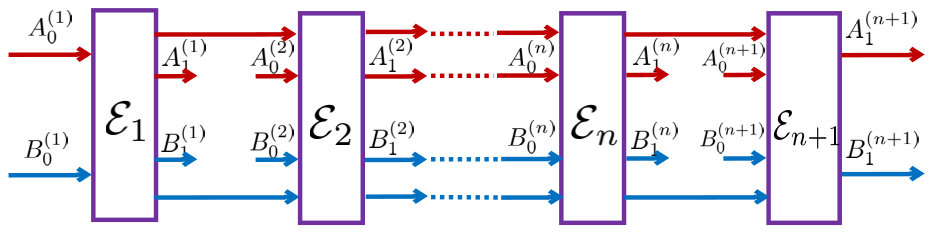}
\par\end{centering}
\caption{\label{fig:pptq}A quantum comb is PPT if and only if the bipartite
channel $\mathcal{Q}\in\mathrm{CPTP}\left(A_{0}^{\left(1\right)}B_{0}^{\left(1\right)}\dots A_{0}^{\left(n+1\right)}B_{0}^{\left(n+1\right)}\rightarrow A_{1}^{\left(1\right)}B_{1}^{\left(1\right)}\dots A_{1}^{\left(n+1\right)}B_{1}^{\left(n+1\right)}\right)$
is PPT.}
\end{figure}

Now we argue that not even such a comb can convert $n$ PPT bipartite
channels $\mathcal{N}_{1},\mathcal{N}_{2},\dots,\mathcal{N}_{n}$
to a single 2-qubit entangled state. This in particular demonstrates
that $n$ adaptive uses of a PPT channel $\mathcal{N}_{AB}$ in a
PPT comb cannot produce a 2-qubit entangled state. In other words,
pure-state entanglement cannot be distilled by LOCC (not even by PPT
combs) out of PPT bipartite channels. In other words, PPT entangled
channels are bound entangled channels.

For this purpose, we note that a comb $\mathscr{C}_{n}$ is uniquely
characterized by the channel $\mathcal{Q}_{A^{n+1}B^{n+1}}$. Therefore,
we define the partial transpose of $\mathscr{C}_{n}$, denoted $\mathscr{C}_{n}^{\Gamma}$,
to be the supermap associated with $\mathcal{Q}_{A^{n+1}B^{n+1}}^{\Gamma}$.
Consequently, $\mathscr{C}$ is a PPT quantum comb if $\mathscr{C}_{n}^{\Gamma}$
is a quantum comb. Moreover, note that
\[
\left(\mathscr{C}_{n}\left[\mathcal{N}_{1},\ldots,\mathcal{N}_{n}\right]\right)^{\Gamma}=\mathscr{C}_{n}^{\Gamma}\left[\mathcal{N}_{1}^{\Gamma},\ldots,\mathcal{N}_{n}^{\Gamma}\right],
\]
as described in Fig.~\ref{fig:pptformula} for $n=2$.
\begin{figure}
\begin{centering}
\includegraphics[width=1\columnwidth]{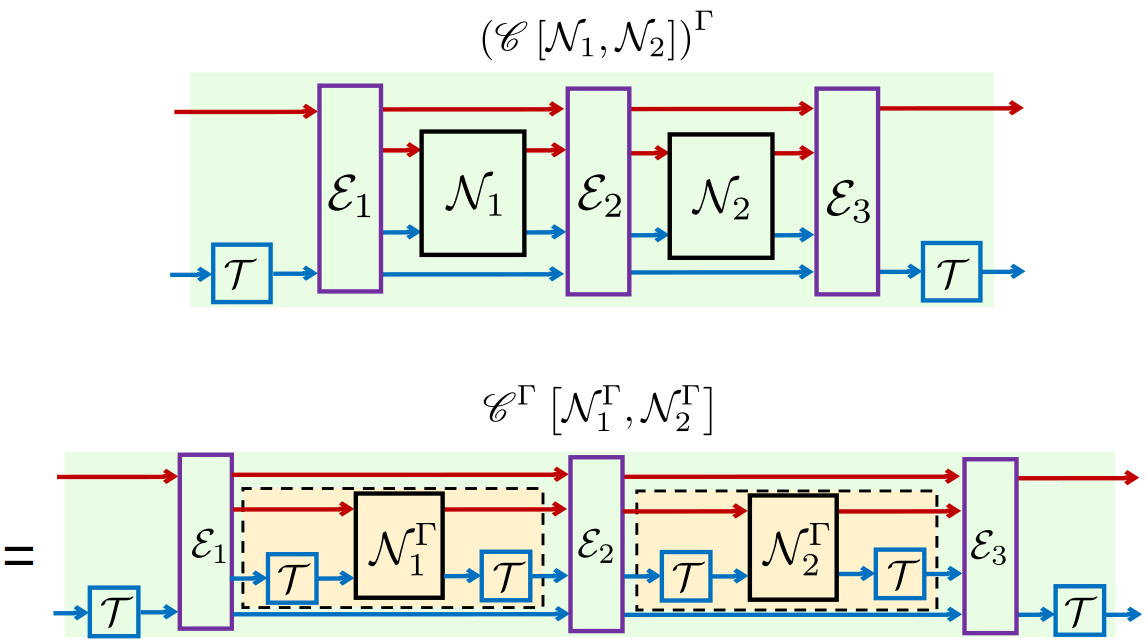}
\par\end{centering}
\caption{\label{fig:pptformula}The channel $\left(\mathscr{C}_{n}\left[\mathcal{N}_{1},\mathcal{N}_{2}\right]\right)^{\Gamma}$
equals the channel $\mathscr{C}_{n}^{\Gamma}\left[\mathcal{N}_{1}^{\Gamma},\mathcal{N}_{2}^{\Gamma}\right]$.
Note that the yellow boxed areas correspond to the maps $\left(\mathcal{N}_{1}^{\Gamma}\right)^{\Gamma}=\mathcal{N}_{1}$
and $\left(\mathcal{N}_{2}^{\Gamma}\right)^{\Gamma}=\mathcal{N}_{2}$.}

\end{figure}
 This is the key reason why PPT quantum combs cannot produce entangled
states from PPT channels.
\begin{prop}
Let $\mathscr{C}_{n}$ be a PPT quantum comb with $n$ slots, as illustrated
in Fig.~\ref{fig:pptcomb}, with $\left|A_{0}^{\left(1\right)}\right|=\left|B_{0}^{\left(1\right)}\right|=1$
and $\left|A_{1}^{\left(n+1\right)}\right|=\left|B_{1}^{\left(n+1\right)}\right|=2$.
Let $\mathcal{N}_{1},\dots,\mathcal{N}_{n}$ be $n$ PPT bipartite
channels with input and output dimensions compatible with the comb
$\mathscr{C}_{n}$, i.e.\ such that $\mathscr{C}_{n}\left[\mathcal{N}_{1},\dots,\mathcal{N}_{n}\right]=:\rho_{A_{1}^{\left(n+1\right)}B_{1}^{\left(n+1\right)}}$
is a well-defined 2-qubit state. Then, the quantum state $\rho_{A_{1}^{\left(n+1\right)}B_{1}^{\left(n+1\right)}}$
is separable.
\end{prop}

\begin{proof}
The proof follows from the property that 
\begin{align*}
\left(\rho_{A_{1}^{\left(n+1\right)}B_{1}^{\left(n+1\right)}}\right)^{T_{B_{1}^{(n+1)}}} & =\left(\mathscr{C}_{n}\left[\mathcal{N}_{1},\dots,\mathcal{N}_{n}\right]\right)^{\Gamma}\\
 & =\mathscr{C}_{n}^{\Gamma}\left[\mathcal{N}_{1}^{\Gamma},\dots,\mathcal{N}_{n}^{\Gamma}\right]\\
 & \geq0
\end{align*}
because $\mathscr{C}_{n}^{\Gamma}$ is a quantum comb, and $\mathcal{N}_{1}^{\Gamma},\dots,\mathcal{N}_{n}^{\Gamma}$
are all CP maps, as $\mathcal{N}_{1},\dots,\mathcal{N}_{n}$ are PPT.
So $\rho_{A_{1}^{\left(n+1\right)}B_{1}^{\left(n+1\right)}}$ is a
PPT 2-qubit state, therefore separable \citep{PPT-Peres,PPT-Horodecki}.
\end{proof}
Note that the above proposition generalizes the notion of bound entanglement
to multiple and possibly different dynamical resources. In the special
case in which $\mathcal{N}_{1}=\dots=\mathcal{N}_{n}\equiv\mathcal{N}$,
the above proposition implies that pure state entanglement cannot
be distilled from a PPT bipartite channel $\mathcal{N}$, not even
with adaptive schemes. When $\mathcal{N}$ has trivial input, we recover
the PPT bound entanglement for states. When $\mathcal{N}\in\mathrm{CPTP}\left(A_{0}B_{0}\to A_{1}B_{1}\right)$
has classical outputs $A_{1}$ and $B_{1}$, we get bound entanglement
for POVMs. Since the latter is a less studied one, we give here a
simple example of a family of bipartite POVMs that are not local (i.e.\ cannot
be implemented by LOCC), but at the same time they cannot produce
distillable entanglement. To find other candidates for bound entangled
channels, we must consider PPT channels that are not LOCC.
\begin{example}
Let $\beta\in\mathfrak{D}\left(A_{0}B_{0}\right)$ be any PPT bound
entangled state of a composite system $A_{0}B_{0}$, and define a
binary POVM consisting of $E_{A_{0}B_{0}}:=\beta_{A_{0}B_{0}}$ and
$F_{A_{0}B_{0}}:=I_{A_{0}B_{0}}-\beta_{A_{0}B_{0}}$. We view this
POVM as the bipartite channel $\mathcal{E}\in\mathrm{CPTP}\left(A_{0}B_{0}\to X\right)$
(as already noted in section~\ref{subsec:Entanglement-POVM}, since
the output is classical, there is no need to represent it with two
classical systems, because classical communication is free) given
by 
\begin{align*}
\mathcal{E}_{A_{0}B_{0}\to X}\left(\rho_{A_{0}B_{0}}\right) & :=\mathrm{Tr}\left[E_{A_{0}B_{0}}\rho_{A_{0}B_{0}}\right]\left|0\right\rangle \left\langle 0\right|_{X}\\
 & +\mathrm{Tr}\left[F_{A_{0}B_{0}}\rho_{A_{0}B_{0}}\right]\left|1\right\rangle \left\langle 1\right|_{X}.
\end{align*}
Since both $E_{A_{0}B_{0}}$ and $F_{A_{0}B_{0}}$ have positive partial
transpose, it follows that $\mathcal{E}$ above is a PPT channel,
and, as such, it cannot produce distillable entanglement. This means
that the POVM $\left\lbrace E_{A_{0}B_{0}},F_{A_{0}B_{0}}\right\rbrace $
is a bound entangled POVM.
\end{example}

\section{Conclusions and outlook\label{sec:Conclusions-and-outlook}}

In this article we studied quantum entanglement as a \emph{resource
theory of processes}, where the resources are bipartite channels (see
Fig.~\ref{bipartitechannel}). This paradigm encompasses several
interesting cases, including the already well-studied resource theory
of entanglement of quantum states \citep{Plenio-review,Review-entanglement},
but also the novel area of entanglement theory for POVMs.

The LOCC resource theory for dynamical entanglement is still very
complicated to characterize from a mathematical point of view, so
we also considered broader classes of free superchannels: separable
superchannels (SEPS) \citep{SEP,Rains-SEP,PPT1} in section~\ref{sec:SEP}
and PPT superchannels \citep{PPT1,PPT2,Leung} in section~\ref{sec:PPT}.
The NPT resource theory is particularly simple to deal with, as all
resource-theoretic protocols can be fully characterized by SDPs. This
remarkable fact, which did not appear in a previous work on PPT superchannels
\citep{WW18}, is a consequence of not restricting ourselves to freely
realizable \citep{Gour-Scandolo-resource} PPT superchannels, i.e.\ to
superchannels whose pre- and post-processing are both PPT channels.
This is not the only novelty with respect to Ref.~\citep{WW18}:
since we considered the most general case of bipartite channels, we
were able to generalize their notion of $\kappa$-entanglement in
two possible ways, the maximum of which we call \emph{max-logarithmic
negativity} (see section~\ref{subsec:NPT-entanglement-measures}).
This has a nice operational characterization as the exact asymptotic
entanglement cost of a bipartite channel under PPT superchannels.

Finally, we concluded with an analysis of bound entanglement for bipartite
channels, showing that from a PPT channel we can distill \emph{no}
ebits under any PPT superchannels (therefore also under any LOCC or
SEP superchannels), not even with an adaptive scheme. This generalizes
the known result for PPT states \citep{Bound-entanglement}. We were
also able to give an example of a bound entangled POVM (section~\ref{sec:Bound}).

Clearly our work just looks at the surface of a whole unexplored world,
but it opens the way to the study of the new area of entanglement
of bipartite channels \citep{Dynamical-coherence-entanglement,Chen2020entanglement,Wilde-biteleportation,Lami-entanglement,Takagi-one-shot,Yuan-one-shot,Kim-Lee,Wilde-entanglement-2ndlaw}.
On a small level, one can generalize the analysis we did, and the
results we obtained in this article. For example, one can try to characterize
which PPT superchannels are freely realizable, i.e.\ restricted PPT
channels (see section~\ref{subsec:PPT-supermaps}), and what the
resulting resource theory looks like. One can also go a level up in
complexity, and describe transitions under LOCC superchannels.

Possible easy directions for future work involve expanding our preliminary
treatment of the entanglement of POVMs (section~\ref{subsec:Entanglement-POVM}
to deal with concrete cases and examples, e.g.\ von Neumann measurements);
studying the entanglement of bipartite unitary channels \citep{Bennett-bipartite},
or even achieving a complete characterization of the entanglement
of the simplest instances of bipartite channels, i.e.\ those where
every system is a qubit. Moreover, another interesting research direction
is about witnesses. In appendix~\ref{sec:PPT-witnesses} we introduce
witnesses for PPT superchannels, but, as we note therein, the really
interesting ones are for the LOCC theory, which have yet to be characterized.

On a grand scale, this work on entanglement theory leads to several
areas that can be explored anew. Think, e.g.\ of multipartite entanglement
\citep{Review-entanglement}, or of the whole zoo of entanglement
measures \citep{Plenio-review,Review-entanglement}. One can also
wonder if entangled bipartite channels can be used to draw a secret
key from them \citep{Ekert}. Moreover, our results for LOCC superchannels
can be translated to local operations and shared randomness (LOSR)
superchannels \citep{Rosset-distributed,Wolfe2020,Schmid2020,LOSR-nonlocality},
which are a strict subset of LOCC ones. LOSR superchannels were argued
to be essential for the formulation of resource theories for non-locality
\citep{LOSR-nonlocality}, as they define the relevant notion of dynamical
entanglement in Bell and common-cause scenarios. This intriguing research
direction deserves a comprehensive study in the future, in addition
to theories of non-locality that do not involve LOSR channels \citep{Gour-nonlocality}.
Finally, providing us with a more general angle, research developments
in the resource theory of entanglement for bipartite channels can
also help us get insights into one of the major open problems of quantum
information theory: the existence of bound entangled states that are
not PPT states \citep{NPPT1,NPPT2,NPPT3}.

To conclude, on an even more general and speculative level, one can
introduce a resource theory of entanglement for higher-level generalizations
of quantum channels \citep{Hierarchy-combs,Perinotti1,Perinotti2},
such as superchannels themselves, combs, or more exotic objects without
a definite causal structure \citep{Switch,Process-matrix,Giarmatzi}.
On such a general level, especially without causal definiteness, one
can expect new and interesting behaviors, which need to be examined
carefully.
\begin{acknowledgments}
G.\ G.\ would like to thank Francesco Buscemi, Eric Chitambar, Mark
Wilde, and Andreas Winter for many useful discussions related to the
topic of this paper. The authors acknowledge support from the Natural
Sciences and Engineering Research Council of Canada (NSERC) through
grant RGPIN-2020-03938, from the Pacific Institute for the Mathematical
Sciences (PIMS), and a from Faculty of Science Grand Challenge award
at the University of Calgary. The authors would like to thank Gaurav
Saxena for a careful reading of an earlier version of this work.
\end{acknowledgments}

\bibliographystyle{apsrev4-2}
\bibliography{PPT}

\appendix
\onecolumngrid

\section{NPT witnesses\label{sec:PPT-witnesses}}

Entanglement witnesses provide a simple ``no-go'' testing to determine
whether a given resource (state, channel, or even superchannel) is
free or not. Here we analyze the witnesses determining whether a bipartite
superchannel is PPT or not, for this is the most general case. Indeed,
PPT states and PPT channels can be viewed as limiting cases of PPT
superchannels when some of the input systems are trivial. 

In NPT entanglement theory one can determine whether bipartite states,
channels, or superchannels are PPT simply by checking the positivity
of their partial transpose. Why do we study NPT witnesses then? The
main reason is to distinguish them from LOCC entanglement witnesses,
which are the more interesting ones. Indeed, LOCC entanglement witnesses
play a central role in entanglement theory, as there is no simple
or efficient way to determine if a resource is entangled or not. Therefore,
distinguishing LOCC witnesses from NPT ones is necessary to understand
which witnesses are truly physically meaningful.
\begin{defn}
A matrix $W\in\mathrm{Herm}\left(ABA'B'\right)$ is an \emph{NPT witness}
if it is \emph{not} positive semi-definite, and if it satisfies $\mathrm{Tr}\left[W_{ABA'B'}\mathbf{J}_{ABA'B'}^{\Theta}\right]\geq0$
for all superchannels in $\Theta\in\mathrm{PPT}\left(AB\to A'B'\right)$.
\end{defn}

Therefore, the set of all NPT witnesses can be viewed as the set of
all matrices in $\mathfrak{J}_{ABA'B'}^{*}$ that are not positive
semi-definite, where $\mathfrak{J}_{ABA'B'}^{*}$ is the dual of the
cone generated by the Choi matrices of PPT bipartite superchannels,
$\mathfrak{J}_{ABA'B'}$. In Ref.~\citep{Gour-Scandolo-resource},
we showed that the former can be expressed as
\begin{equation}
\mathfrak{J}_{ABA'B'}^{*}=\left\{ W\in\mathrm{Herm}\left(ABA'B'\right):\mathrm{Tr}\left[W_{ABA'B'}\mathbf{J}_{ABA'B'}^{\Theta}\right]\geq0\right\} ,\label{eq:dual cone PPT}
\end{equation}
for every PPT superchannel $\Theta$. Then $\mathfrak{J}_{ABA'B'}^{*}$
is the set of all $W\in\mathrm{Herm}\left(ABA'B'\right)$ such that
$\mathrm{Tr}\left[W_{ABA'B'}\mathbf{J}_{ABA'B'}\right]\geq0$ for
all matrices $\mathbf{J}\in\mathrm{Herm}\left(ABA'B'\right)$ with
the following properties:
\begin{enumerate}
\item $\mathbf{J}_{ABA'B'}\geq0$;
\item $\mathbf{J}_{ABA'_{0}B'_{0}}=\mathbf{J}_{A_{0}B_{0}A_{0}'B_{0}'}\otimes u_{A_{1}B_{1}}$;
\item $\mathbf{J}_{A_{1}B_{1}A'_{0}B'_{0}}=I_{A_{1}B_{1}A_{0}'B_{0}'}$;
\item $\mathbf{J}_{ABA'B'}^{T_{BB'}}\geq0$.
\end{enumerate}
Note that the first three conditions ensure that $\mathbf{J}_{ABA'B'}$
is the Choi matrix of a bipartite superchannel, and the last condition
ensures that the superchannel is PPT.

The conditions above imply that all NPT witnesses $W\in\mathrm{Herm}\left(ABA'B'\right)$
are of the form
\begin{equation}
W_{ABA'B'}=P_{ABA'B'}+X_{ABA'B'}^{T_{BB'}}+Y_{ABA_{0}'B_{0}'}\otimes I_{A_{1}'B_{1}'}+I_{A_{0}B_{0}A_{1}'B_{1}'}\otimes Z_{A_{1}B_{1}A_{0}'B_{0}'},\label{eq:form}
\end{equation}
where $P_{ABA'B'},X_{ABA'B'}\geq0$, $Y_{ABA_{0}'B_{0}'}$ is a Hermitian
matrix such that $Y_{AB}=0$, and $Z$ is a Hermitian matrix such
that $\mathrm{Tr}\left[Z_{A_{1}B_{1}A_{0}'B_{0}'}\right]=0$. Note
that the Hilbert-Schmidt inner product between $Y_{ABA_{0}'B_{0}'}\otimes I_{A_{1}'B_{1}'}$
(or $I_{A_{0}B_{0}A_{1}'B_{1}'}\otimes Z_{A_{1}B_{1}A_{0}'B_{0}'}$)
and any Choi matrix of a superchannel is always zero, as shown in
Ref.~\citep{Gour2018}. This is why they can be added to any NPT
witness. Now we will use this form of NPT witnesses to expresses the
PPT conversion distance as an SDP.

\section{Additivity of the max-logarithmic negativity\label{app:MLN}}

Here we prove only the additivity of $LN_{\max}^{\left(0\right)}$,
as the proof of the additivity of $LN_{\max}^{\left(1\right)}$ follows
the exact same lines.
\begin{lem}
For any two bipartite channels $\mathcal{N}\in\mathrm{CPTP}\left(A_{0}B_{0}\to A_{1}B_{1}\right)$
and $\mathcal{M}\in\mathrm{CPTP}\left(A_{0}'B_{0}'\to A_{1}'B_{1}'\right)$
we have 
\[
LN_{\max}^{\left(0\right)}\left(\mathcal{N}_{AB}\otimes\mathcal{M}_{A'B'}\right)=LN_{\max}^{\left(0\right)}\left(\mathcal{N}_{AB}\right)+LN_{\max}^{\left(0\right)}\left(\mathcal{M}_{A'B'}\right).
\]
\end{lem}

\begin{proof}
For simplicity of the exposition, in some places we will omit the
subscripts identifying the systems. By definition we have 
\begin{align}
LN_{\max}^{\left(0\right)}\left(\mathcal{N}_{AB}\otimes\mathcal{M}_{A'B'}\right) & =\log_{2}\inf\left\{ \left\Vert J_{A_{0}B_{0}A_{0}'B_{0}'}^{\mathcal{P}}\right\Vert _{\infty}:-\mathcal{P}_{ABA'B'}^{\Gamma}\leq\mathcal{N}_{AB}^{\Gamma}\otimes\mathcal{M}_{A'B'}^{\Gamma}\leq\mathcal{P}_{ABA'B'}^{\Gamma},\,\mathcal{P}\geq0\right\} \label{cc1}\\
 & \leq\log_{2}\inf\left\{ \left\Vert J_{A_{0}B_{0}A_{0}'B_{0}'}^{\mathcal{P}_{1}\otimes\mathcal{P}_{2}}\right\Vert _{\infty}:-\mathcal{P}_{1}^{\Gamma}\leq\mathcal{N}^{\Gamma}\leq\mathcal{P}_{1}^{\Gamma};\thinspace-\mathcal{P}_{2}^{\Gamma}\leq\mathcal{M}^{\Gamma}\leq\mathcal{P}_{2}^{\Gamma};\thinspace\mathcal{P}_{1},\mathcal{P}_{2}\geq0\right\} \label{cc2}\\
 & =LN_{\max}^{\left(0\right)}\left(\mathcal{N}_{AB}\right)+LN_{\max}^{\left(0\right)}\left(\mathcal{M}_{A'B'}\right),\nonumber 
\end{align}
where the inequality follows from the fact that, if $\mathcal{P}_{1}$
and $\mathcal{P}_{2}$ satisfy the constraints in~\eqref{cc2}, then
$\mathcal{P}=\mathcal{P}_{1}\otimes\mathcal{P}_{2}$ satisfies the
constraints in~\eqref{cc1}. The last equality follows from the multiplicativity
of the operator norm under tensor product.

For the other direction, we use the dual expression in Eq.~\eqref{eq:MLN0}.
Hence,
\[
LN_{\max}^{\left(0\right)}\left(\mathcal{N}_{AB}\otimes\mathcal{M}_{A'B'}\right)=\log_{2}\sup\left\{ \mathrm{Tr}\left[J^{\mathcal{N}\otimes\mathcal{M}}\left(V-W\right)\right]:V+W\leq\rho\otimes I;\thinspace\rho\in\mathfrak{D}\left(A_{0}B_{0}A_{0}'B_{0}'\right);\thinspace V,W\geq0\right\} .
\]
 Setting $X:=V+W$ and $Y:=V-W$, we have
\begin{align}
LN_{\max}^{\left(0\right)}\left(\mathcal{N}_{AB}\otimes\mathcal{M}_{A'B'}\right) & =\log_{2}\sup\left\{ \mathrm{Tr}\left[J^{\mathcal{N}\otimes\mathcal{M}}Y\right]:X\leq\rho\otimes I;\,\rho\in\mathfrak{D}\left(A_{0}B_{0}A_{0}'B_{0}'\right);\thinspace X\pm Y\geq0\right\} \label{dd1}\\
 & \geq\log_{2}\sup\left\{ \mathrm{Tr}\left[J^{\mathcal{N}\otimes\mathcal{M}}\left(Y_{1}\otimes Y_{2}\right)\right]:X_{1}\leq\rho_{1}\otimes I;\thinspace X_{2}\leq\rho_{2}\otimes I;\thinspace X_{1}\pm Y_{1}\geq0;\thinspace X_{2}\pm Y_{2}\geq0\right\} ,\label{dd2}
\end{align}
where $\rho_{1}\in\mathfrak{D}\left(A_{0}B_{0}\right)$ and $\rho_{2}\in\mathfrak{D}\left(A'_{0}B'_{0}\right)$
and the inequality follows from the fact that if $X_{1},X_{2},\rho_{1},\rho_{2}$
satisfy the constraints in~\eqref{dd2}, then $X=X_{1}\otimes X_{2}$,
$Y=Y_{1}\otimes Y_{2}$, and $\rho=\rho_{1}\otimes\rho_{2}$ satisfy
the constraints in~\eqref{dd1}. In particular, let us show that
if $X_{1}\pm Y_{1}\geq0$ and $X_{2}\pm Y_{2}\geq0$, then we also
have $X_{1}\otimes X_{2}\pm Y_{1}\otimes Y_{2}\geq0$. First of all,
observe that, from the assumptions $X_{1}\pm Y_{1}\geq0$ and $X_{2}\pm Y_{2}\geq0$,
we have 
\[
\left(X_{1}\pm Y_{1}\right)\otimes\left(X_{2}\pm Y_{2}\right)\geq0,
\]
from which 
\[
X_{1}\otimes X_{2}+Y_{1}\otimes Y_{2}\geq\mp\left(X_{1}\otimes Y_{2}+Y_{1}\otimes X_{2}\right).
\]
This means that 
\[
\left\langle \psi\middle|X_{1}\otimes X_{2}+Y_{1}\otimes Y_{2}\middle|\psi\right\rangle \geq\mp\left\langle \psi\middle|X_{1}\otimes Y_{2}+Y_{1}\otimes X_{2}\middle|\psi\right\rangle ,
\]
for all vectors $\psi$. This in turn means that 
\[
\left\langle \psi\middle|X_{1}\otimes X_{2}+Y_{1}\otimes Y_{2}\middle|\psi\right\rangle \geq0,
\]
for all vectors $\psi$, from which $X_{1}\otimes X_{2}+Y_{1}\otimes Y_{2}\geq0$.

Similarly, from 
\[
\left(X_{1}\pm Y_{1}\right)\otimes\left(X_{2}\mp Y_{2}\right)\geq0
\]
we get that
\[
X_{1}\otimes X_{2}-Y_{1}\otimes Y_{2}\geq\mp\left(Y_{1}\otimes X_{2}-X_{1}\otimes Y_{2}\right),
\]
which, by an argument similar to the one above, allows us to conclude
that $X_{1}\otimes X_{2}-Y_{1}\otimes Y_{2}\geq0$.

Combining both inequalities we obtained for $LN_{\max}^{\left(0\right)}\left(\mathcal{N}_{AB}\otimes\mathcal{M}_{A'B'}\right)$,
we prove the additivity.
\end{proof}

\end{document}